\documentclass[letterpaper,11pt]{article}

\usepackage{bm}
\usepackage{url}
\usepackage{array}
\usepackage{color}
\usepackage{amsthm}
\usepackage{dsfont}
\usepackage{amsmath}
\usepackage{amssymb}
\usepackage{authblk}
\usepackage{caption}
\usepackage{amsfonts}
\usepackage{booktabs}
\usepackage{graphicx}
\usepackage{mathrsfs}
\usepackage{multicol}
\usepackage{multirow}
\usepackage{cellspace}
\usepackage{enumerate}
\usepackage{enumitem}
\usepackage{makecell}

\usepackage{hyperref}
\hypersetup{
	linktoc      = all,
	colorlinks   = true,
	urlcolor     = blue,
	linkcolor    = blue,
	citecolor    = red
}

\usepackage[backend=biber, isbn=false, style=alphabetic, backref=true, doi=true, url=false, maxcitenames=10, mincitenames=10, maxalphanames=10, maxbibnames=10, minbibnames=5, minalphanames=5, defernumbers=true, sortlocale=en_US]{biblatex}
\usepackage[T1]{fontenc}
\usepackage[dvipsnames]{xcolor}
\usepackage[margin=1in]{geometry}
\usepackage[capitalize,noabbrev]{cleveref}
\usepackage[ruled,vlined,linesnumbered]{algorithm2e}
\usepackage[caption=false,font=normalsize,labelfont=sf,textfont=sf]{subfig}

\usepackage{titlesec}
\titlespacing*{\paragraph}{0pt}{1ex}{1ex}

\setlist[itemize]{leftmargin=*}
\setlist[enumerate]{leftmargin=*}
\setlist[enumerate]{label=(\arabic*)}

\newtheorem{theorem}{Theorem}[section]
\newtheorem{lemma}[theorem]{Lemma}

\newtheorem{corollary}[theorem]{Corollary}
\newtheorem{definition}[theorem]{Definition}

\crefname{algocf}{Algorithm}{Algorithms}
\Crefname{algocf}{Algorithm}{Algorithms}

\crefname{algocfline}{Line}{Lines}
\Crefname{algocfline}{Line}{Lines}

\crefname{lemma}{Lemma}{Lemmas}
\Crefname{lemma}{Lemma}{Lemmas}

\DeclareMathOperator{\polylog}{polylog}
\DeclareMathOperator{\nnz}{nnz}
\DeclareMathOperator{\vol}{vol}
\DeclareMathOperator{\supp}{supp}
\DeclareMathOperator{\E}{\mathbb{E}}
\DeclareMathOperator{\Var}{Var}

\DeclareMathOperator{\tvol}{\widetilde{vol}}

\def\eps{\varepsilon}

\def\R{\mathbb{R}}

\def\tO{\widetilde{O}}

\def\mat{\mathbf}

\def\A{\mat{A}}
\def\D{\mat{D}}
\def\I{\mat{I}}
\def\mL{\mat{L}}
\def\M{\mat{M}}

\def\S{\mat{S}}

\def\vec{\boldsymbol}

\def\b{\vec{b}}
\def\d{\vec{d}}
\def\e{\vec{e}}
\def\p{\vec{p}}
\def\r{\vec{r}}
\def\s{\vec{s}}
\def\x{\vec{x}}

\def\vzero{\vec{0}}
\def\vone{\vec{1}}

\def\pr{\vec{\mathrm{pr}}}

\def\tp{\tilde{\p}}
\def\tx{\tilde{\x}}
\def\tr{\tilde{\r}}

\def\xast{\x^{\ast}}
\def\Sast{S^{\ast}}

\def\PageRankNibble{\textup{\texttt{PageRank-Nibble}}\xspace}
\def\activeSet{\textup{\texttt{ActiveSetPageRank}}\xspace}
\def\generalActiveSet{\textup{\texttt{GeneralActiveSetPageRank}}\xspace}
\def\sparseActiveSet{\textup{\texttt{SparseActiveSetPageRank}}\xspace}
\def\accumulate{\textup{\texttt{Accumulate}}\xspace}

\title{Accelerating the Local Push Primitive for PageRank Computation}

\author{Guanyu Cui}
\author{Zhewei Wei}
\author{Mingji Yang}

\affil{Renmin University of China \authorcr
	\{cuiguanyu,zhewei,kyleyoung\}@ruc.edu.cn}

\date{}

\begin{document}

\maketitle

\begin{abstract}

We propose a local algorithm that computes an $\eps$-approximate PageRank vector in the sense of Andersen, Chung, and Lang (ACL; Internet Math. 2007) with teleportation parameter $\alpha$ in $\tO\bigl(1 / \bigl(\sqrt{\alpha} \, \eps\bigr)\bigr)$ time with high probability, improving the $O\bigl(1/(\alpha\eps)\bigr)$ running time of their original local push method.
Our method also applies to the $\ell_1$-regularized PageRank problem with a running time of $\tO\bigl(1 / \bigl(\sqrt{\alpha} \, \rho\bigr)\bigr)$ for regularization parameter $\rho$, giving a positive answer to the open problem posed by Fountoulakis and Yang (COLT 2022).

Our faster primitive has the potential to improve a broad range of graph algorithms that rely on local push.
For example, substituting our primitive into the ACL framework directly yields faster PageRank-based local graph clustering, and we also develop reductions that lead to faster algorithms for effective resistance estimation.

Our main technical contribution is a potential-function analysis of a refinement of the active-set method of Wei and Yang (preprint 2026), which repeatedly invokes an SDD solver on the current active set of nodes and expands the set.
We relate the potential decreases over consecutive blocks of expansions to show that the number of expansions is bounded by $\tO\bigl(1 / \sqrt{\alpha}\bigr)$.

\end{abstract}

\section{Introduction} \label{sec:intro}

PageRank~\cite{brin1998anatomy} is a fundamental measure of node centrality or proximity with a natural interpretation in terms of random-walk probabilities.
Given a teleportation parameter $\alpha \in (0,1)$, the PageRank value from a seed node $s$ to a node $v$ equals the probability that a random walk that starts from $s$ terminates at $v$, where the walk terminates with probability $\alpha$ before each step.
Beyond its original use in web ranking, PageRank has been widely used in network analysis, data mining~\cite{gleich2015pagerank,yang2024efficient}, and theoretical computer science.
As a basic concept and tool in spectral graph theory, it has been applied to numerous topics in theoretical computer science, including local graph clustering~\cite{andersen2006local,andersen2007pagerank,andersen2007detecting,zhu2013local,orecchia2014flow}, directed Laplacian solvers~\cite{cohen2016faster}, small-set expansion~\cite{chan2017random,kwok2017improved}, edge connectivity~\cite{kawarabayashi2019deterministic}, and unique games~\cite{yoshida2026tolerant}.
Additionally, local computation of PageRank has itself received sustained attention, with a long line of work and several recent advances~\cite{andersen2007pagerank,andersen2008local,lofgren2014fast,lofgren2015bidirectional,lofgren2016personalized,wang2020personalized,bressan2023sublinear,wei2024approximating,wang2024revisiting,wang2024revisitinga,bertram2025estimating,thorup2026pagerank,thorup2026instance,bertram2026personalized}.

A central algorithmic primitive in PageRank computation and its applications is the \textit{local push} method of Andersen, Chung, and Lang~\cite{andersen2006local,andersen2007pagerank}, originally developed for improving the local graph clustering algorithm by Spielman and Teng~\cite{spielman2004nearly}.
The local push method approximates the PageRank vector from a seed node by repeatedly distributing residual probability mass from a node to its neighbors, and its locality and approximation guarantees make it a useful building block for other graph algorithms.
It is the core ingredient in PageRank-based graph clustering and partitioning~\cite{andersen2006local,andersen2007pagerank,andersen2007detecting,zhu2013local,orecchia2014flow,kawarabayashi2019deterministic}, and along with its reverse variant and extensions, it has been widely used in local PageRank estimation~\cite{andersen2008local,lofgren2015bidirectional,lofgren2016personalized,wang2020personalized,wei2024approximating,wang2024revisiting,bertram2025estimating,thorup2026pagerank,thorup2026instance,bertram2026personalized}, effective resistance estimation~\cite{cui2025mixing,yang2025improved}, linear-system solving~\cite{kwok2026solving}, and estimating other quantities~\cite{banerjee2015fast,wang2021approximate,bressan2023sublinear}.
These applications motivate us to improve the efficiency of the local push primitive itself.

We study the local push primitive on undirected graphs.
Given a seed node $s$, a teleportation parameter $\alpha \in (0,1)$, and an accuracy parameter $\eps > 0$, the ACL push algorithm~\cite{andersen2007pagerank} computes an ACL $\eps$-approximate PageRank vector, defined formally in \cref{def:ACL} below, in $O\bigl(1/(\alpha\eps)\bigr)$ time, while both the support volume of the vector and the query complexity are $O(1 / \eps)$.
A natural question is whether the linear dependence on $1 / \alpha$ in the time complexity can be improved while preserving the $1 / \eps$ dependence.
Such an improved primitive has the potential to yield faster algorithms in a broad range of applications that rely on the local push method.
Recently, Wei and Yang~\cite{wei2026simple} proposed an active-set algorithm that computes an ACL $\eps$-approximate PageRank vector in $\tO\bigl(1 / \eps^2\bigr)$ time, which achieves only a polylogarithmic dependence on $1 / \alpha$, but the dependence on $1 / \eps$ becomes quadratic.

On the other hand, a related line of work studies PageRank estimation from an optimization perspective.
\cite{fountoulakis2019variational} formulated approximate PageRank computation through $\ell_1$-regularization and established the locality of a proximal-gradient method, and Fountoulakis and Yang~\cite{fountoulakis2022open} then posed the open problem of whether accelerated methods could improve the running time from $\tO\bigl(1/(\alpha\rho)\bigr)$ to $\tO\bigl(1/\bigl(\sqrt{\alpha}\,\rho\bigr)\bigr)$ for regularization parameter $\rho$.
Subsequently, \cite{rubio2023accelerated} achieved the $1 / \sqrt{\alpha}$ dependence but with a larger dependence on $1 / \rho$.
More recently, Fountoulakis and Mart\'inez-Rubio~\cite{fountoulakis2026complexity} established strong hardness results for applying a classic accelerated proximal-gradient method to this problem, but this does not rule out the possibility of achieving the $\tO\bigl(1/(\sqrt{\alpha}\,\rho)\bigr)$ time complexity using general algorithms.

In this paper, we show that an ACL $\eps$-approximate PageRank vector with support volume $O(1 / \eps)$ can be computed in $\tO\bigl(1 / \bigl(\sqrt{\alpha}\,\eps\bigr)\bigr)$ time with high probability, improving the $O\bigl(1/(\alpha\eps)\bigr)$ running time of the original local push primitive.
We also show that an additive approximate minimizer of the $\ell_1$-regularized PageRank problem can be computed in $\tO\bigl(1 / \bigl(\sqrt{\alpha}\,\rho\bigr)\bigr)$ time, answering the open problem of Fountoulakis and Yang~\cite{fountoulakis2022open} in the affirmative.
Our algorithm is a refinement of the active-set method of \cite{wei2026simple} with a new potential-function analysis showing that the number of iterations in the active-set expansion process is $\tO\bigl(1 / \sqrt{\alpha}\bigr)$.

In the rest of this section, we formally define the problems, introduce our main results and the applications, and provide a technical overview of our approach.
We also pose an open problem of whether the $\tO\bigl(1 / \bigl(\sqrt{\alpha}\,\eps\bigr)\bigr)$ time complexity can be further improved to $\tO(1 / \eps)$.

\subsection{Problem Formulation}

Throughout, $G = (V,E)$ denotes a connected, unweighted, undirected simple graph on $n := |V| \ge 2$ nodes.
We write $\A$ for its adjacency matrix, $\D$ for the diagonal matrix of node degrees, and $\d$ for the degree vector.
The neighborhood of a node $v \in V$ is denoted by $N(v)$, and $\e_v \in \R^V$ is the standard basis vector associated with $v$.
For a set $S \subseteq V$, its volume and internal volume are defined by $\vol(S) := \sum_{v \in S}\d(v)$ and $\tvol(S) := |S| + \bigl|\bigl\{(u,v) \in E: u \in S, v \in S\bigr\}\bigr|$, respectively.
For a vector $\x \in \R^V$, we use $\supp(\x) := \bigl\{v : \x(v) \ne 0\bigr\}$ to denote its support.
Our local algorithms access the graph through the standard adjacency-list model: querying a node's degree or retrieving an individual neighbor takes $O(1)$ time.

For any source vector $\s \in \R^V$ and a \textit{teleportation parameter} $\alpha \in (0,1)$, we use $\pr_{\alpha}(\s) \in \R^V$ to denote its PageRank vector, which is defined as the unique solution to the linear system
\begin{align}
	\pr_{\alpha}(\s) = \alpha\s + (1-\alpha) \A\D^{-1} \pr_{\alpha}(\s). \label{eqn:PageRank}
\end{align}
As a shorthand, we write $\pr_{\alpha}(u,v) := \pr_{\alpha}(\e_u)(v)$ for the PageRank value from $u$ to $v$.
Note that we do not regard $\alpha$ as a constant in this paper.
We focus on the following sense of approximate PageRank vector as considered by \cite{andersen2007pagerank,andersen2007detecting}.

\begin{definition}[ACL $\eps$-approximate PageRank vector~\cite{andersen2007pagerank,andersen2007detecting}] \label[definition]{def:ACL}
	For a given source vector $\s \in \R^V$ and a teleportation parameter $\alpha \in (0,1)$, a vector $\p \in \R^V_{\ge 0}$ is called an $\eps$-approximate PageRank vector if there exists a vector $\r$ such that $\p = \pr_{\alpha}(\s-\r)$ and $\vzero \le \r \le \eps \d$ entrywise.
\end{definition}

When we only specify a seed node $s \in V$, we take the source vector as $\e_s$ by default.
Our main goal is to compute an ACL $\eps$-approximate PageRank vector with respect to a seed node $s$ with support volume $O(1 / \eps)$.

We also study the following $\ell_1$-regularized PageRank problem, which is equivalent to the ones in previous works~\cite{fountoulakis2019variational,fountoulakis2022open,rubio2023accelerated} up to variable scaling and parameter transformation.

\paragraph{$\ell_1$-regularized PageRank.}
Given $G$, $s \in V$, $\alpha$, and a regularization parameter $\rho > 0$, the $\ell_1$-regularized PageRank problem~\cite{fountoulakis2019variational} is to minimize
\begin{align}
	\psi_{\rho}(\x) := \frac{1}{2} \, \x^{\top} \bigl(\D - (1 - \alpha)\A\bigr) \x - \alpha\,\e_s^{\top}\x + \alpha\rho \, \|\D\x\|_1 \label{eqn:regularized-objective}
\end{align}
over $\x \in \R^V$.
It is shown in \cite{fountoulakis2019variational} that $\psi_{\rho}(\cdot)$ has a unique minimizer $\xast \ge \vzero$ such that $\D\xast$ is an ACL $\rho$-approximate PageRank vector and $\vol\bigl(\supp(\xast)\bigr) \le 1 / \rho$.
We study the problem of computing a $\xi$-additive approximate minimizer of $\psi_{\rho}(\cdot)$ for a given $\xi > 0$, i.e., computing a vector $\tx \in \R^V$ such that $\psi_{\rho}(\tx) - \psi_{\rho}(\xast) \le \xi$.

\subsection{Our Main Results}

Our main result is a randomized algorithm that computes an ACL $\eps$-approximate PageRank vector in $\tO\bigl(1 / \bigl(\sqrt{\alpha} \, \eps\bigr)\bigr)$ time with high probability, stated as follows.

\begin{theorem} \label{thm:ACL}
	Given a seed node $s \in V$, $\alpha \in (0,1)$, $\eps > 0$, and $\delta \in (0,1)$, there exists an algorithm that satisfies the following properties with probability at least $1 - \delta$:
	\begin{itemize}
		\item it returns an ACL $\eps$-approximate PageRank vector with support volume at most $2/\eps$;
		\item its running time is
		\begin{align*}
			O\Bigl( \min\Bigl\{\frac{1}{\sqrt{\alpha} \, \eps},\frac{1}{\eps^2}\Bigr\} \polylog\frac{1}{\alpha\eps\delta} \Bigr).
		\end{align*}
	\end{itemize}
\end{theorem}

In comparison, the original local push method~\cite{andersen2007pagerank} computes an ACL $\eps$-approximate PageRank vector in deterministic $O\bigl(1/(\alpha\eps)\bigr)$ time with support volume $O(1 / \eps)$.
The polylogarithmic factors in the running time of our algorithm rely on the invoked SDD (symmetric diagonally dominant) solvers~\cite{spielman2014nearly} in our algorithms and can be small using improved  solvers~\cite{jambulapati2025ultrasparse}.
We remark that the $\tO\bigl( 1 / \eps^2 \bigr)$ time complexity bound for this task is already obtained by \cite{wei2026simple} and is also naturally achieved by our algorithm.

Using the same approach, we also obtain results for computing an additive approximate minimizer of the $\ell_1$-regularized PageRank problem in \eqref{eqn:regularized-objective}.

\begin{theorem} \label{thm:regularized}
	Given a seed node $s \in V$, $\alpha \in (0,1)$, $\rho > 0$, $\xi > 0$, and $\delta \in (0,1)$ and letting $\xast$ be the unique minimizer for $\psi_{\rho}(\cdot)$ and $\Sast := \supp(\xast)$, there exists an algorithm that satisfies the following properties with probability at least $1 - \delta$:
	\begin{itemize}
		\item it returns $\tx$ such that $\psi_{\rho}(\tx)-\psi_{\rho}(\xast) \le \xi$, $\D\tx$ is an ACL $(2\rho)$-approximate PageRank vector, and $\vol\bigl(\supp(\tx)\bigr) \le 1 / \rho$;
		\item its running time is bounded by both
		\begin{align*}
			O\Bigl(\min\Bigl\{\frac{1}{\sqrt{\alpha} \, \rho},\frac{1}{\rho^2}\Bigr\} \polylog\frac{1}{\alpha\rho\xi\delta}\Bigr)
		\end{align*}
		and
		\begin{align*}
			O\Bigl( \frac{1}{\sqrt{\alpha}} \tvol(\Sast) \polylog\frac{1}{\alpha\rho\xi\delta} + |\Sast| \vol(\Sast) \Bigr).
		\end{align*}
	\end{itemize}
\end{theorem}

The $\tO\bigl(1/\bigl(\sqrt{\alpha}\,\rho\bigr)\bigr)$ time complexity bound above answers the open problem of \cite{fountoulakis2022open} in the affirmative.
Additionally, the ASPR algorithm in \cite{rubio2023accelerated} computes a $\xi$-additive approximate minimizer in deterministic time
\begin{align*}
	O\Bigl( \frac{1}{\sqrt{\alpha}} \, |\Sast| \tvol(\Sast) \log\frac{1}{\alpha\rho\xi} + |\Sast| \vol(\Sast) \Bigr).
\end{align*}
Compared to this bound, up to polylogarithmic factors, our result removes a factor of $|\Sast|$ in the first term.

Similar to \cite{wei2026simple}, the randomness of the algorithms in \cref{thm:ACL,thm:regularized} only comes from the calls to randomized SDD solvers (cf. \cref{sec:overview}).
Replacing these subroutines with deterministic almost-linear-time SDD solvers~\cite{chuzhoy2020deterministic} yields deterministic algorithms whose running-time bounds no longer depend on the failure probability $\delta$.
The resulting bounds incur additional factors of $(1 / \eps)^{o(1)}$ and $|\Sast|^{o(1)} \le (1 / \rho)^{o(1)}$, respectively.

Through a change of the teleportation parameter, our results also apply to a variant of PageRank defined using lazy random walks.
Specifically, the lazy-walk definition with parameter $\alpha_{\mathrm L} = \alpha / (2 - \alpha)$ gives the same PageRank vector as the non-lazy definition with parameter $\alpha$~\cite{andersen2007pagerank}.
Since $1 / \alpha_{\mathrm L} = (2 - \alpha) / \alpha = \Theta(1 / \alpha)$, this conversion preserves the asymptotic dependence on the inverse teleportation parameter in all our bounds.

\subsection{Applications of the Improved Local Push Primitive}

We apply our improved local push primitive to the following three graph problems.

\paragraph{Local graph clustering.}
Directly replacing local push in the \PageRankNibble algorithm of Andersen, Chung, and Lang~\cite{andersen2007pagerank} with our primitive improves its running-time dependence on the inverse target conductance.
We present this application in \cref{sec:PageRank-Nibble}.

\paragraph{Single-pair PageRank estimation.}
The same substitution in the bidirectional PageRank estimator of Lofgren, Banerjee, and Goel~\cite{lofgren2015bidirectional} improves the running-time dependence on $1 / \alpha$ for single-pair PageRank estimation.
We give the details in \cref{sec:single-pair-PageRank}.

\paragraph{Effective resistance estimation.}
We design a new reduction that estimates single-pair effective resistance through a short sum of PageRank matrix powers, which can be approximated using repeated calls to the local push primitive combined with random-walk sampling.
The reduction improves existing bidirectional algorithms~\cite{cui2025mixing,yang2025improved,kwok2026solving} even with the original local push method, and plugging in our faster primitive yields a further improvement.
We present the reduction and its analysis in \cref{sec:effective-resistance}.

\subsection{Technical Overview} \label{sec:overview}

We focus here on computing an ACL $\eps$-approximate PageRank vector with respect to a seed node $s$.
For this task, the original local push method maintains nonnegative reserve and residue vectors $\p$ and $\r$ that satisfy the invariant property $\p = \pr_{\alpha}(\e_s-\r)$ throughout.
The vectors $\p$ and $\r$ should be understood as probability mass on the nodes of the graph, and initially $\p = \vzero$ and $\r = \e_s$.
The local push method repeatedly performs push operations on nodes with excessive residues, where each push operation transfers $\alpha$ fraction of the residue of the chosen node to its reserve and propagates the remaining $(1 - \alpha)$ fraction of the residue to its neighbors.
The algorithm terminates when the residue vector satisfies $\r \le \eps \d$ entrywise, at which time the accumulated reserve vector $\p$ is an ACL $\eps$-approximate PageRank vector by definition.

The method in \cite{wei2026simple} organizes this mass-propagation process differently as a sequence of active-set expansions.
It maintains an active set $S$ of nodes, and conceptually performs an infinite process that only performs push operations on the nodes in the set $S$ until convergence, while only the excess residue above the threshold $0.5 \eps \, \d(v)$ is pushed.
After this infinite push process, only the nodes on the external boundary of $S$ may receive excessive residues.
If any of them indeed violates the residue threshold required by an ACL $\eps$-approximate PageRank vector, the algorithm expands $S$ by adding these boundary nodes to it and repeats the process.

Crucially, for each encountered active set $S$, the limiting state of the infinite push process on $S$ can be approximated by solving an SDD system on $S$.
The algorithm ensures that the computed approximate states are sufficiently accurate to guarantee that the active-set expansion process preserves the monotonicity and nonnegativity properties of a pure push-based process.
Also, we can guarantee that the volume of the active set $S$ is $O(1 / \eps)$ throughout, and thus each SDD solve takes $\tO(1 / \eps)$ time using nearly-linear-time SDD solvers~\cite{spielman2014nearly}.
\cite{wei2026simple} directly bounds the number of active-set expansions by $\vol(S) = O(1 / \eps)$, which yields a total running time of $\tO\bigl(1 / \eps^2\bigr)$.

We follow the active-set framework of \cite{wei2026simple}, but refine the expansion criterion by separating the activation and stopping thresholds.
We still stop when all the residues on the external boundary of $S$ are at most $\eps \, \d(v)$, but otherwise activate every boundary node whose residue exceeds a threshold slightly above $0.5 \eps \, \d(v)$.
Intuitively, this avoids the situation where a boundary node with residue slightly below $\eps \, \d(v)$ is not activated in one iteration, but its residue exceeds $\eps \, \d(v)$ in the next iteration, necessitating an additional expansion of the active set to include it.

Specifically, we analyze the refined active-set process by using the $\ell_1$-regularized objective $\psi_{0.5 \eps}(\cdot)$ of the degree-normalized reserve vector as the potential function of the process.
Our key lemma compares adjacent blocks of $\Theta\bigl(1 / \sqrt{\alpha}\bigr)$ expansions, showing that the potential decrease over the second block is at most a constant fraction of that over the first, plus a small additive term that comes from the small activation gap and solver error.
Intuitively, this reflects the diminishing benefit of further expansions, as measured by the potential decrease over successive blocks.
We establish the lemma by analyzing the potential decrease in the second block of expansions in terms of the residue changes in the first block.
This result implies that the potential decrease per block decays almost geometrically and would fall below the minimum decrease required by a complete block of non-terminating iterations after only logarithmically many blocks.
Thus, the number of iterations is bounded by $\tO(1 / \sqrt{\alpha})$, which yields the total running time of $\tO\bigl(1 / \bigl(\sqrt{\alpha} \, \eps\bigr)\bigr)$.
This algorithmic framework and analysis also extend to $\ell_1$-regularized PageRank.

\subsection{Other Related Work}

Most work on local PageRank computation studies the estimation of individual PageRank values or vectors under absolute or relative error guarantees~\cite{lofgren2015bidirectional,lofgren2016personalized,bressan2023sublinear,wei2024approximating,wang2024revisiting,bertram2025estimating,bertram2026personalized}.
Moreover, many of these works treat the teleportation parameter $\alpha$ as a constant and focus on the dependence on graph parameters.
Thus, the approximation guarantees and parameter regimes that they study are different from ours.

There are also works on local solvers for SDD systems~\cite{andoni2019solving} and asymmetric diagonally dominant systems~\cite{kwok2026solving}.
However, these solvers rely on techniques that are commonly used in local PageRank computation, and thus their results cannot be applied or adapted to yield a faster algorithm for our purposes.

PageRank estimation has also been studied in multi-pass graph streams~\cite{sarma2011estimating} and single-pass random-order streams~\cite{kallaugher2022simulating}, using random-walk simulation to obtain space-efficient algorithms.
Other work considers computing PageRank in the Massively Parallel Computation (MPC) model~\cite{lacki2020walking} and maintaining approximate PageRank under edge updates in dynamic graphs~\cite{jayaram2024dynamic}.

\subsection{Future Directions}

A natural open problem is whether an ACL $\eps$-approximate PageRank vector with support volume $O(1 / \eps)$ can be computed in $\tO(1 / \eps)$ time, with only a polylogarithmic dependence on $1 / \alpha$.
The original local push algorithm already achieves both support volume and query complexity $O(1 / \eps)$, and, to our knowledge, no hardness result rules out the running-time bound of $\tO(1 / \eps)$ for general algorithms.
Such an improvement would further accelerate applications of local push, and in particular, it would remove the polynomial dependence on the reciprocal of the desired conductance in the running time of \PageRankNibble (cf. \cref{cor:PageRank-Nibble}) for local graph clustering.
From a broader perspective, this open problem reflects the challenge of narrowing the gap between query complexity and computational complexity in sublinear-time algorithms, which is generally underexplored as highlighted by recent work on property testing~\cite{ferreira2026computational}.

We also believe that our improved local push primitive has more applications beyond those explored in this paper.
It would be interesting to identify other graph algorithms that can benefit from the accelerated local push primitive.

\subsection{Paper Organization}

The rest of the paper is organized as follows.
In \cref{sec:prelim}, we introduce notation and basic tools.
In \cref{sec:alg}, we present the active-set algorithm, analyze the number of expansions, and establish our main results.
\cref{sec:applications} gives the applications to local graph clustering and single-pair PageRank estimation.
Finally, \cref{sec:effective-resistance} develops the reduction from effective resistance estimation to PageRank computation and the resulting faster algorithms.

\section{Preliminaries} \label{sec:prelim}

For any $\x \in \R^V$ and $S \subseteq V$, we use both $\x_S \in \R^S$ and $\x|_S \in \R^S$ to denote the restriction of $\x$ to the coordinates in $S$.
When we use $\le$ and $\ge$ between two vectors, we mean entrywise comparison.
We use $\vzero$ and $\vone$ to denote the all-zeros and all-ones vectors, respectively, with appropriate dimensions.

For $S \subseteq V$, let matrices $\A_S$ and $\D_S$ be the principal submatrices of $\A$ and $\D$ indexed by $S$, respectively.
Define the Laplacian matrix as $\mL := \D - \A$, and define $\mL_{\alpha} := \D - (1 - \alpha) \A$ and $\mL_{\alpha,S} := \D_S - (1 - \alpha) \A_S$.
We call a square matrix \textit{symmetric diagonally dominant} (SDD) if it is symmetric and each of its diagonal entries is at least the sum of the absolute values of the other entries in the same row.
It follows that $\mL_{\alpha,S}$ is SDD for any $S \subseteq V$.
Also, define the external boundary of $S \subseteq V$ by $\partial S:= \bigl\{v \in V \setminus S : N(v) \cap S \ne \varnothing\bigr\}$.

For a vector $\x$ and a symmetric positive definite matrix $\S$, define the energy norm as $\|\x\|_{\S} := \sqrt{\x^{\top} \S \x} = \bigl\|\S^{1/2}\x\bigr\|_2$.
For a matrix $\M$, let $\M^{+}$ denote its Moore--Penrose pseudoinverse.
The effective resistance between two distinct nodes $s,t \in V$ is defined as $R(s,t) := (\e_s - \e_t)^{\top} \mL^{+} (\e_s - \e_t)$.

In our algorithm, we use nearly-linear-time approximate solvers for SDD linear systems as a black box.
We state below a version of SDD solvers with a deterministic running-time bound, which can be obtained from existing results by standard repetition and truncation operations.

\begin{theorem}[SDD solver~\cite{spielman2014nearly,jambulapati2025ultrasparse}] \label{thm:SDD_solver}
	There is a randomized algorithm $\textup{\texttt{SDDSolve}}(\S, \b, \mu, \delta)$ that, given a positive definite SDD matrix $\S \in \R^{N \times N}$, a vector $\b \in \R^N$, an accuracy parameter $\mu \in (0,1)$, and a failure probability $\delta \in (0,1)$, runs in $O\bigl( \nnz(\S) \polylog\bigl(N / (\mu\delta)\bigr) \bigr)$ time and with probability at least $1 - \delta$ returns a vector $\tx \in \R^N$ such that
	\begin{align}
		\bigl\| \tx - \S^{-1}\b \bigr\|_{\S} \le \mu \bigl\| \S^{-1}\b \bigr\|_{\S}. \label{eqn:solver-guarantee}
	\end{align}
\end{theorem}

Nearly-linear-time solvers for SDD systems originate in the work of Spielman and Teng~\cite{spielman2004nearly} and have since then been improved and simplified in a long line of work.
Note that the term $\polylog\bigl(N / (\mu\delta)\bigr)$ in the complexity bound above can be rather small (e.g., using \cite{jambulapati2025ultrasparse}), but we leave it implicit in this paper.

\subsection{Properties of PageRank}

For every $u \in V$, the PageRank values satisfy $\pr_{\alpha}(u,v) \ge 0$ for all $v \in V$ and $\sum_{v \in V} \pr_{\alpha}(u,v) = 1$.
On undirected graphs, they also satisfy the following symmetry identity~\cite{arvachenkov2013choice,lofgren2015bidirectional}, which is a consequence of the reversibility of random walks on undirected graphs:
\begin{align}
	\d(u) \, \pr_{\alpha}(u,v) = \d(v) \, \pr_{\alpha}(v,u), \qquad \forall \, u,v \in V. \label{eqn:symmetry}
\end{align}
Next, fix a seed node $s \in V$ and consider the invariant property $\p = \pr_{\alpha}(\e_s - \r)$ as required in the definition of an ACL $\eps$-approximate PageRank vector. 
We call $\p$ the \textit{reserve vector} and $\r$ the \textit{residue vector}.
The invariant property is equivalent to
\begin{align}
	\pr_{\alpha}(s,v) = \p(v) + \sum_{u \in V} \r(u) \, \pr_{\alpha}(u,v), \qquad \forall \, v \in V. \label{eqn:invariant-entrywise}
\end{align}
The following notation and properties are given in \cite{wei2026simple}.
For any $\x \in \R^V$ (which should be thought of as the degree-normalized reserve vector $\D^{-1} \p$), we define its associated residue vector $\r(\x) \in \R^V$ as
\begin{align}
	\r(\x) := \e_s - \frac{1}{\alpha} \, \mL_{\alpha} \, \x. \label{eqn:residue-def}
\end{align}
Note that given $\x \in \R^V$, we can compute $\r(\x)$ in $O\bigl(\vol\bigl(\supp(\x)\bigr)\bigr)$ time by scanning the edges incident to $\supp(\x)$.
For any $\x \in \R^V$, $\D\x$ and $\r(\x)$ satisfy the invariant $\D\x = \pr_{\alpha}\bigl(\e_s - \r(\x)\bigr)$, and we also have
\begin{align}
	\mL_{\alpha} \, \x = \alpha \bigl(\e_s - \r(\x)\bigr). \label{eqn:residue-invariant}
\end{align}
Thus, given a fixed residue vector $\r(\x)$, the corresponding degree-normalized reserve vector $\x$ can be recovered by solving the linear system \eqref{eqn:residue-invariant}.
Furthermore, if we are given $\r \in \R^S$ and want to compute a vector $\x \in \R^V$ such that $\r(\x)_S = \r$ and $\x_{V \setminus S} = \vzero$, we can set $\x_S$ by solving the restricted linear system $\mL_{\alpha,S} \, \x_S = \alpha (\e_s|_S - \r)$.
Here, $\mL_{\alpha,S}$ is SDD and invertible.

The definition of $\r(\x)$ also guarantees that, for any $\x \in \R^V$,
\begin{align}
	\vone^{\top} \D \x + \vone^{\top} \r(\x) = 1. \label{eqn:mass-identity}
\end{align}
If $\supp(\x) \subseteq S$ and $s \in S$ for some $S \subseteq V$, then
\begin{align}
	\r(\x)(v) = \frac{1 - \alpha}{\alpha} \sum_{u \in N(v) \cap S} \x(u), \qquad \forall \, v \in V \setminus S. \label{eqn:frontier-residue-formula}
\end{align}
In particular, $\r(\x)(v) = 0$ for every $v \notin S \cup \partial S$.

\section{The Active-Set Algorithm} \label{sec:alg}

In this section, we present our active-set algorithm for computing approximate PageRank vectors.

\subsection{Algorithm Description and Basic Properties}

We first describe the algorithm for computing an ACL $\eps$-approximate PageRank vector.
Following the framework of \cite{wei2026simple}, our algorithm alternates between approximately solving an SDD system on the current active set and expanding this set according to the residues on the external boundary of the active set.
We give pseudocode in \cref{alg:active-set}.
If $\eps \, \d(s) \ge 1$, the algorithm immediately returns $\tp = \vzero$, which satisfies the ACL guarantee with residue $\r = \e_s$.
Otherwise, it initializes the active set to $S := \{s\}$ and sets the activation gap $\kappa := 0.1 \alpha \eps^2$ and solver accuracy $\mu := 0.01 \alpha \kappa$.

In each iteration, the algorithm invokes the SDD solver from \cref{thm:SDD_solver} to approximately solve the restricted system $\mL_{\alpha,S} \, \x_S = \alpha (\e_s|_S - 0.5 \eps \, \d_S)$.
Recall that the exact solution to this system is the degree-normalized reserve vector supported on $S$ whose corresponding residue vector equals $0.5 \eps \, \d_S$ on $S$.
The solver returns an approximation $\tx_S$ with accuracy $\mu$ and failure probability at most $0.5\eps\delta$.
The algorithm then computes the residues on the external boundary $\partial S$ by \cref{eqn:frontier-residue-formula}.
If $\tr(v) \le \eps \, \d(v)$ for every $v \in \partial S$, then $\D_S \, \tx_S$ extended to $V$ by zero is a valid ACL $\eps$-approximate PageRank vector and the algorithm returns it.
Otherwise, it adds every boundary node satisfying $\tr(v) > (0.5 \eps + \kappa) \, \d(v)$ to $S$ and repeats the solve on the enlarged active set.

Note that in the algorithm, the activation threshold $0.5\eps + \kappa$ is lower than the stopping threshold $\eps$, so an expansion may also include nodes whose residues already satisfy the stopping condition.
On the other hand, as $\kappa < 0.5 \eps$ provided that $\eps \, \d(s) < 1$, each non-terminating iteration adds at least one node to the active set $S$.

\begin{algorithm}[ht]
	\DontPrintSemicolon
	\caption{$\activeSet(s,\alpha,\eps,\delta)$} \label{alg:active-set}
	\KwIn{seed node $s \in V$, teleportation parameter $\alpha \in (0,1)$, accuracy $\eps > 0$, failure probability $\delta \in (0,1)$}
	\KwOut{an ACL $\eps$-approximate PageRank vector $\tp$ with probability at least $1 - \delta$}
	\If{$\eps \, \d(s) \ge 1$}
	{
		\Return{$\tp \gets \vzero$} \;
	}
	\makebox[3cm][l]{$\kappa \gets 0.1 \alpha \eps^2$} \textcolor{gray}{// activation gap} \;
	\makebox[3cm][l]{$\mu \gets 0.01 \alpha \kappa$} \textcolor{gray}{// SDD solver accuracy} \;
	$S \gets \{s\}$ \;
	\While{\textup{true}}
	{
		$\tx_S \gets \textup{\texttt{SDDSolve}}\bigl(\mL_{\alpha,S}, \alpha (\e_s|_S - 0.5 \eps \, \d_S), \mu, 0.5 \eps \delta\bigr)$ \;
		$\tr(v) \gets \frac{1 - \alpha}{\alpha} \sum_{u \in N(v) \cap S} \tx_S(u)$ for each $v \in \partial S$ \;
		\If{$\tr(v) \le \eps \, \d(v)$ \textup{for all} $v \in \partial S$}
		{
			\Return{$\tp$ \textup{with} $\tp_S \gets \D_S \, \tx_S$ \textup{and} $\tp_{V \setminus S} \gets \vzero$} \;
		}
		$S \gets S \cup \bigl\{v \in \partial S: \tr(v) > (0.5 \eps + \kappa) \, \d(v)\bigr\}$ \;
	}
\end{algorithm}

For ease of later adaptation to solving $\ell_1$-regularized PageRank, we will analyze a more general version of the active-set algorithm given in \cref{alg:general-active-set}.
The algorithm takes as input an internal residue threshold $\lambda > 0$, a termination gap $\eta > 0$, and a small activation gap $0 < \kappa \le 0.4 \alpha \eta^2$.
This algorithm sets the SDD accuracy parameter as $\mu := 0.01 \alpha \min\{\lambda,\kappa\}$, and solves the restricted systems and performs the expansion according to the parameters $\lambda$, $\eta$, and $\kappa$.
Clearly, \cref{alg:active-set} is the instantiation of \cref{alg:general-active-set} with $\lambda = \eta = 0.5 \eps$ and $\kappa = 0.1 \alpha \eps^2$.

\begin{algorithm}[ht]
	\DontPrintSemicolon
	\caption{$\generalActiveSet(s,\alpha,\lambda,\kappa,\eta,\delta)$} \label{alg:general-active-set}
	\KwIn{seed node $s \in V$, teleportation parameter $\alpha \in (0,1)$, internal threshold $\lambda > 0$, termination gap $\eta > 0$, activation gap $0 < \kappa \le 0.4 \alpha \eta^2$, failure probability $\delta \in (0,1)$}
	\KwOut{an approximate PageRank vector $\tp$}
	\If{$(\lambda + \eta) \, \d(s) \ge 1$}
	{
		\Return{$\tp \gets \vzero$} \;
	}
	$\mu \gets 0.01 \alpha \min\{\lambda,\kappa\}$ \quad \textcolor{gray}{// SDD solver accuracy} \;
	$S \gets \{s\}$ \;
	\While{\textup{true}}
	{
		$\tx_S \gets \textup{\texttt{SDDSolve}}\bigl(\mL_{\alpha,S}, \alpha (\e_s|_S - \lambda \, \d_S), \mu, \lambda\delta\bigr)$ \;
		$\tr(v) \gets \frac{1 - \alpha}{\alpha} \sum_{u \in N(v) \cap S} \tx_S(u)$ for each $v \in \partial S$ \;
		\If{$\tr(v) \le (\lambda + \eta) \, \d(v)$ \textup{for all} $v \in \partial S$}
		{
			\Return{$\tp$ \textup{with} $\tp_S \gets \D_S \, \tx_S$ \textup{and} $\tp_{V \setminus S} \gets \vzero$} \;
		}
		$S \gets S \cup \bigl\{v \in \partial S: \tr(v) > (\lambda + \kappa) \, \d(v)\bigr\}$ \;
	}
\end{algorithm}

To analyze \cref{alg:general-active-set}, we first introduce some additional notation.
For a set $S$ containing $s$, define $\x^{(S)} \in \R^V$ by
\begin{align}
	\x^{(S)}_S := \alpha \, \mL_{\alpha,S}^{-1} (\e_s|_S - \lambda \, \d_S), \qquad \x^{(S)}_{V \setminus S} := \vzero. \label{eqn:exact-solution}
\end{align}
Plugging \cref{eqn:exact-solution} into \cref{eqn:residue-def} verifies that $\r\bigl(\x^{(S)}\bigr)_S = \lambda \, \d_S$.
Let $S_0,S_1,\ldots$ be the active sets encountered by \cref{alg:general-active-set}.
At every reached iteration $i$, extend the solver output (the vector $\tx_S$ in the pseudocode) by zero outside $S_i$, denote the resulting vector by $\tx_i$, and define $\tp_i := \D\tx_i$ and $\tr_i := \r(\tx_i)$.
Also, associate with the same active set the exact state $\x_i := \x^{(S_i)}$, $\p_i := \D\x_i$, and $\r_i := \r(\x_i)$.
Whenever iteration $i$ performs an expansion, define its expansion set by
\begin{align*}
	T_i := \bigl\{v \in \partial S_i : \tr_i(v) > (\lambda + \kappa) \, \d(v)\bigr\} = S_{i + 1} \setminus S_i.
\end{align*}

We first reuse some basic properties of the active-set process from \cite{wei2026simple}.
The following lemma formalizes its monotonicity property.

\begin{lemma}[{\cite[Lemma~4.2]{wei2026simple}}] \label[lemma]{lem:active-set-expansion}
	Suppose $s \in S \subseteq V$ and $T \subseteq \partial S$ satisfies $\r\bigl(\x^{(S)}\bigr)_T > \lambda \, \d_T$.
	Then $\x^{(S \cup T)} \ge \x^{(S)}$ and
	\begin{align}
		\x^{(S \cup T)}(v) \ge \frac{\alpha}{\d(v)}\bigl(\r\bigl(\x^{(S)}\bigr)(v) - \lambda \, \d(v)\bigr), \qquad \forall \, v \in T. \label{eqn:active-set-expansion}
	\end{align}
\end{lemma}

The next lemma bounds the entrywise error in $\tx$ and $\tr$ introduced by one successful approximate SDD solve.

\begin{lemma}[{\cite[Lemma~4.3]{wei2026simple}}] \label[lemma]{lem:active-set-solver-error}
	Suppose that $(\lambda + \eta) \, \d(s) < 1$ and in an iteration of \cref{alg:general-active-set}, the set $S$ satisfies $\x^{(S)} \ge \vzero$.
	Suppose that the approximate SDD solve in this iteration returns a $\tx_S$ satisfying the guarantee \eqref{eqn:solver-guarantee}.
	Extend $\tx_S$ to $\tx \in \R^V$ by setting $\tx_{V \setminus S} := \vzero$, and let $\tr := \r(\tx)$.
	Then
	\begin{alignat}{2}
		\bigl|\tr(v) - \lambda \, \d(v)\bigr| & \le 0.1 \min\{\lambda,\kappa\} \, \d(v), & \qquad & \forall \, v \in S, \label{eqn:approx-active-residue-error} \\
		\bigl|\tx_S(v) - \x^{(S)}(v)\bigr| & \le 0.1 \alpha \kappa, & \qquad & \forall \, v \in S, \label{eqn:approx-coordinate-error} \\
		\bigl|\tr(v) - \r\bigl(\x^{(S)}\bigr)(v)\bigr| & \le 0.1 \kappa \, \d(v), & \qquad & \forall \, v \in \partial S. \label{eqn:approx-frontier-residue-error}
	\end{alignat}
\end{lemma}

Let $\mathcal{G}$ be the ``good'' event that every approximate SDD call invoked by \cref{alg:general-active-set} satisfies the guarantee \eqref{eqn:solver-guarantee}.
The following lemma lower-bounds the probability of $\mathcal{G}$ and establishes the basic properties of the active-set process on this event.
The proof adapts the argument in \cite[Lemma~4.4]{wei2026simple} to our two-threshold rule.

\begin{lemma} \label[lemma]{lem:trajectory}
	We have $\Pr[\mathcal{G}] \ge 1 - \delta$.
	On $\mathcal{G}$, \cref{alg:general-active-set} terminates after at most $1 / \lambda$ iterations and returns $\tp \ge \vzero$ with $\vzero \le \r\bigl(\D^{-1} \tp\bigr) \le (\lambda + \eta)\d$.
	Moreover, on $\mathcal{G}$, for every iteration $i$, we have $\x_i,\p_i,\r_i,\tx_i,\tp_i,\tr_i \ge \vzero$, $\x_i(v) > 0.9 \alpha \kappa$ for every $v \in S_i$, and $\vol(S_i) \le 1 / \lambda$.
	On $\mathcal{G}$, whenever the stopping rule is not met at iteration $i$, we have $\x_{i + 1} \ge \x_i$, and
	\begin{align}
		\begin{cases}
			\r_i(v) > (\lambda + 0.9 \kappa) \, \d(v), & \qquad \forall \, v \in T_i, \\
			\r_i(v) \le (\lambda + 1.1 \kappa) \, \d(v), & \qquad \forall \, v \notin S_{i+1}.
		\end{cases}
		\label{eqn:residue-bounds}
	\end{align}
\end{lemma}

We note that the bound $\r_i(v) \le (\lambda + 1.1 \kappa) \, \d(v)$ for $v \notin S_{i+1}$ stated above will be crucial later.
It guarantees that the residues of the nodes outside the next active set can only be slightly larger than the internal threshold $\lambda \, \d(v)$.

\begin{proof}[Proof of \cref{lem:trajectory}]
	Consider the nontrivial case that $(\lambda + \eta) \, \d(s) < 1$, and first condition on $\mathcal{G}$.
	We prove by induction that $\x_i,\r_i \ge \vzero$ and $\x_i(v) > 0.9 \alpha \kappa$ for every $v \in S_i$.
	Note that by $\r_i|_{S_i} = \lambda \, \d_{S_i}$ and \cref{eqn:frontier-residue-formula}, $\x_i \ge \vzero$ implies that $\r_i \ge \vzero$.
	For $i = 0$, we have $S_0 = \{s\}$ and by \cref{eqn:exact-solution},
	\begin{align*}
		\x_0(s) = \alpha \Bigl(\frac{1}{\d(s)} - \lambda\Bigr) > \alpha \eta \ge \alpha \kappa > 0.9 \alpha \kappa,
	\end{align*}
	so the base case holds.

	Suppose that these properties hold at iteration $i$ and the stopping rule is not met.
	Then \cref{lem:active-set-solver-error} applies, and thus for every $v \in T_i$, the definition of $T_i$ and \eqref{eqn:approx-frontier-residue-error} give $\r_i(v) > (\lambda + \kappa) \, \d(v) - 0.1\kappa \, \d(v) = (\lambda + 0.9 \kappa) \, \d(v)$.
	On the other hand, for $v \in \partial S_i \setminus T_i$, \eqref{eqn:approx-frontier-residue-error} gives $\r_i(v) \le (\lambda + 1.1\kappa) \, \d(v)$, while $\r_i(v) = 0$ for $v \notin S_i \cup \partial S_i$, so we have proved \eqref{eqn:residue-bounds}.
	Applying \cref{lem:active-set-expansion} with $S = S_i$ and $T = T_i$ now gives $\x_{i + 1} \ge \x_i \ge \vzero$.
	For $v \in S_i$, this along with the induction hypothesis gives $\x_{i + 1}(v) \ge \x_i(v) > 0.9 \alpha \kappa$, and for $v \in T_i$, by \eqref{eqn:active-set-expansion}, we have $\x_{i + 1}(v) \ge \frac{\alpha}{\d(v)} \bigl(\r_i(v) - \lambda \, \d(v)\bigr) > 0.9 \alpha \kappa$, as desired.
	This completes the induction.

	It follows that $\p_i = \D \x_i \ge \vzero$.
	The mass identity \cref{eqn:mass-identity} gives $1 \ge \sum_{v \in S_i}\r_i(v) = \lambda \vol(S_i)$, so $\vol(S_i) \le 1 / \lambda$.
	Since $S_0 = \{s\}$ and each expansion adds at least one node, we have $i + 1 \le |S_i| \le \vol(S_i) \le 1 / \lambda$ at every reached iteration $i$.
	Thus, on $\mathcal{G}$, the algorithm terminates after at most $1 / \lambda$ iterations.

	For the approximate vectors, combining the exact bounds above with \cref{lem:active-set-solver-error} gives $\tx_i(v) > 0$ and $\tr_i(v) \ge \bigl(\lambda - 0.1 \min\{\lambda,\kappa\}\bigr) \, \d(v) \ge 0.9 \lambda \, \d(v) > 0$ for every $v \in S_i$.
	Since $\tx_i$ is zero outside $S_i$ and $\tp_i = \D\tx_i$, we obtain $\tx_i,\tp_i \ge \vzero$; \cref{eqn:frontier-residue-formula} then gives $\tr_i(v) \ge 0$ outside $S_i$ as well.

	At termination, \eqref{eqn:approx-active-residue-error} gives $\tr(v) \le (\lambda + 0.1 \kappa) \, \d(v) \le (\lambda + \eta) \, \d(v)$ on the active set.
	The stopping rule gives the same upper bound on its external boundary, and the residue is zero elsewhere by \cref{eqn:frontier-residue-formula}.
	This proves the claimed output guarantee.

	Finally, to bound $\Pr[\mathcal{G}]$, note that the first failed SDD solver invocation must occur at some iteration $j \le 1 / \lambda$.
	For each such iteration $j$, the probability that the first failure occurs at iteration $j$ is at most $\lambda\delta$ by \cref{thm:SDD_solver} and the parameter setting.
	Summing over the possible values of $j$ gives $1 - \Pr[\mathcal{G}] \le \delta$.
\end{proof}

\subsection{Bounding the Number of Active-Set Expansions}

For the remainder of the analysis, assume that $(\lambda + \eta) \, \d(s) < 1$ and condition on the event $\mathcal{G}$.
Suppose that the algorithm performs $K$ active-set expansions before terminating.
Thus, the trajectory defined in the preceding subsection consists of the active sets $S_0 \subsetneq \cdots \subsetneq S_K$, the expansion sets $T_i = S_{i + 1} \setminus S_i$ for $0 \le i < K$, and the exact states $(\x_i,\p_i,\r_i)$ for $0 \le i \le K$.
For every $0 \le i < K$, define
\begin{align*}
	\Delta \x_i := \x_{i + 1} - \x_i, \qquad \Delta \p_i := \p_{i + 1} - \p_i, \qquad \Delta \r_i := \r_{i + 1} - \r_i.
\end{align*}

We will analyze the change of $\psi_{\lambda}(\x_i)$ throughout the iterations.
Recall from \cref{eqn:regularized-objective} that for $\x \ge \vzero$,
\begin{align}
	\psi_{\lambda}(\x) = \frac{1}{2} \, \x^{\top} \mL_{\alpha} \, \x - \alpha \, \x(s) + \alpha \lambda \sum_{v \in V} \d(v) \, \x(v). \label{eqn:objective-quadratic-extension}
\end{align}
Define the one-step and interval objective decreases by
\begin{alignat*}{2}
	\Delta \psi_i & := \psi_{\lambda}(\x_i) - \psi_{\lambda}(\x_{i + 1}), & \qquad & \forall \, 0 \le i < K, \\
	\Delta \psi(a,b) & := \psi_{\lambda}(\x_a) - \psi_{\lambda}(\x_b), & \qquad & \forall \, 0 \le a < b \le K.
\end{alignat*}

The following lemma establishes the one-step objective identity and its lower bound.

\begin{lemma} \label[lemma]{lem:one-step-objective-decrease}
	For every $0 \le i < K$, we have
	\begin{align}
		\Delta \psi_i = -\frac{1}{2} \alpha \sum_{v \in V} \Delta \x_i(v) \, \Delta \r_i(v) = \frac{1}{2} \, \Delta \x_i^{\top} \, \mL_{\alpha} \, \Delta \x_i \label{eqn:one-step-objective-identity}
	\end{align}
	and
	\begin{align}
		\Delta \psi_i \ge \frac{1}{2} \alpha^2 \sum_{v \in T_i} \frac{\Delta \r_i(v)^2}{\d(v)}. \label{eqn:delta-psi-lower-bound}
	\end{align}
\end{lemma}

\begin{proof}
	Fix an iteration $0 \le i < K$.
	By \cref{lem:trajectory}, $\x_i,\x_{i + 1} \ge \vzero$, so applying \cref{eqn:objective-quadratic-extension} yields
	\begin{align*}
		\Delta \psi_i & = \frac{1}{2} \bigl(\x_i^{\top} \, \mL_{\alpha} \, \x_i - \x_{i + 1}^{\top} \, \mL_{\alpha} \, \x_{i + 1}\bigr) + \alpha \, \Delta \x_i(s) - \alpha \lambda \sum_{v \in V} \d(v) \, \Delta \x_i(v) \\
		& = -\frac{1}{2} \, \Delta \x_i^{\top} \, \bigl(\mL_{\alpha}(\x_i + \x_{i + 1}) - 2 \alpha \e_s + 2 \alpha \lambda \d\bigr) \\
		& = -\frac{1}{2} \alpha \sum_{v \in V} \Delta \x_i(v) \bigl(2 \lambda \, \d(v) - \r_i(v) - \r_{i + 1}(v)\bigr),
	\end{align*}
	where the last step follows from $\mL_{\alpha}(\x_i + \x_{i + 1}) = \alpha (2 \e_s - \r_i - \r_{i + 1})$ by \cref{eqn:residue-invariant}.

	For $v \in S_i$, we have $\r_i(v) = \r_{i + 1}(v) = \lambda \, \d(v)$, so $2 \lambda \, \d(v) - \r_i(v) - \r_{i + 1}(v) = 0$.
	For $v \notin S_{i + 1}$, we have $\x_i(v) = \x_{i+1}(v) = \Delta \x_i(v) = 0$.
	For $v \in T_i$, we have $\r_{i+1}(v) = \lambda \, \d(v)$, so
	\begin{align*}
		2 \lambda \, \d(v) - \r_i(v) - \r_{i + 1}(v) = \lambda \, \d(v) - \r_i(v) = \Delta \r_i(v).
	\end{align*}
	Combining these three cases gives
	\begin{align*}
		\Delta \psi_i = -\frac{1}{2} \alpha \sum_{v \in T_i} \Delta \x_i(v) \, \Delta \r_i(v) = -\frac{1}{2} \alpha \sum_{v \in V} \Delta \x_i(v) \, \Delta \r_i(v),
	\end{align*}
	as desired.
	On the other hand, by \cref{eqn:residue-invariant} we have $\mL_{\alpha} \, \Delta \x_i = -\alpha \, \Delta \r_i$.
	Consequently, we further have
	\begin{align*}
		\Delta \psi_i = -\frac{1}{2} \alpha \, \Delta \x_i^{\top} \, \Delta \r_i = \frac{1}{2} \, \Delta \x_i^{\top} \, \mL_{\alpha} \, \Delta \x_i.
	\end{align*}
	This proves the first part of the lemma.

	To prove the lower bound on $\Delta \psi_i$, note that for every $v \in T_i$, \cref{lem:active-set-expansion} and $\x_i(v) = 0$ imply that
	\begin{align*}
		\Delta \p_i(v) = \d(v) \, \Delta \x_i(v) \ge \alpha\bigl(\r_i(v) - \lambda \, \d(v)\bigr) = -\alpha \, \Delta \r_i(v),
	\end{align*}
	where $\Delta \r_i(v) < 0$.
	Combining this with \eqref{eqn:one-step-objective-identity} yields
	\begin{align*}
		\Delta \psi_i = -\frac{1}{2} \alpha \sum_{v \in T_i} \frac{\Delta \p_i(v)}{\d(v)} \, \Delta \r_i(v) \ge \frac{1}{2} \alpha^2 \sum_{v \in T_i} \frac{\Delta \r_i(v)^2}{\d(v)},
	\end{align*}
	as desired.
\end{proof}

The next lemma shows that for $i \ne j$, the vectors $\Delta \x_i$ and $\Delta \r_j$ are orthogonal.

\begin{lemma} \label[lemma]{lem:orthogonal}
	For any $i \ne j$, $\sum_{v \in V} \Delta \x_i(v) \, \Delta \r_j(v) = 0$.
\end{lemma}

\begin{proof}
	First consider the case of $i < j$.
	In this case, we have $\supp(\Delta \x_i) \subseteq S_{i + 1} \subseteq S_j$ and $\Delta \r_j|_{S_j} = \vzero$.
	Therefore, $\sum_{v \in V} \Delta \x_i(v) \, \Delta \r_j(v) = 0$, as desired.

	Next assume that $i > j$.
	By the invariant equation \eqref{eqn:invariant-entrywise}, it holds that, for every $v \in V$,
	\begin{align*}
		\Delta \p_i(v) = - \sum_{u \in V} \Delta \r_i(u) \, \pr_{\alpha}(u,v).
	\end{align*}
	Using this fact and the symmetry identity \eqref{eqn:symmetry}, we have
	\begin{align*}
		\sum_{v \in V} \Delta \x_i(v) \, \Delta \r_j(v) & = \sum_{v \in V} \frac{\Delta \p_i(v)}{\d(v)} \, \Delta \r_j(v) = -\sum_{u \in V} \sum_{v \in V} \frac{1}{\d(v)} \, \Delta \r_i(u) \, \pr_{\alpha}(u,v) \, \Delta \r_j(v) \\
		& = -\sum_{u \in V} \sum_{v \in V} \frac{1}{\d(u)} \, \Delta \r_i(u) \, \pr_{\alpha}(v,u) \, \Delta \r_j(v) = \sum_{u \in V} \Delta \x_j(u) \, \Delta \r_i(u),
	\end{align*}
	so this is also zero by the previous case.
\end{proof}

Combining the one-step identity with the orthogonality lemma gives the following identity for the objective decrease over any interval.

\begin{lemma} \label[lemma]{lem:interval-objective-decrease}
	For every $0 \le a < b \le K$,
	\begin{align}
		\Delta \psi(a,b) = -\frac{1}{2} \alpha \sum_{v \in V} \bigl(\x_b(v) - \x_a(v)\bigr) \bigl(\r_b(v) - \r_a(v)\bigr). \label{eqn:interval-objective-decrease}
	\end{align}
\end{lemma}

\begin{proof}
	Summing \eqref{eqn:one-step-objective-identity} and using \cref{lem:orthogonal} to eliminate all cross terms gives
	\begin{align*}
		\Delta \psi(a,b) & = -\frac{1}{2} \alpha \sum_{i = a}^{b - 1} \sum_{v \in V} \Delta \x_i(v) \, \Delta \r_i(v) \\
		& = -\frac{1}{2} \alpha \sum_{v \in V} \biggl(\sum_{i = a}^{b - 1} \Delta \x_i(v)\biggr) \biggl(\sum_{j = a}^{b - 1} \Delta \r_j(v)\biggr).
	\end{align*}
	The two inner sums above telescope to $\x_b - \x_a$ and $\r_b - \r_a$, respectively, proving the lemma.
\end{proof}

The following key lemma compares the objective decreases over two adjacent blocks of $L$ expansions.

\begin{lemma} \label[lemma]{lem:batch-decrease}
	For any integers $a$ and $L$ such that $a \ge L \ge 1$ and $a + L \le K$,
	\begin{align*}
		\Delta \psi(a,a + L) \le \frac{1}{\alpha L^2} \, \Delta \psi(a - L,a) + 1.1 \alpha \kappa.
	\end{align*}
\end{lemma}

\begin{proof}
	We will bound $\Delta \psi(a,a + L)$ in terms of $\Delta \r_i$ and $\Delta \r_j$ for every pair of $i,j \in \{a - L,\ldots,a - 1\}$, and then sum over all $L^2$ pairs.
	To this end, we will use the following residue properties.

	\paragraph{Residue differences.}
	For any $0 \le i < a+L$, we have
	\begin{align}
		\left\{
		\begin{alignedat}{2}
			\r_i(v) - \r_{a+L}(v) & = 0, & \qquad & \forall \, v \in S_i, \\
			\r_i(v) - \r_{a+L}(v) & = -\Delta \r_i(v), & \qquad & \forall \, v \in T_i, \\
			\r_i(v) - \r_{a+L}(v) & \le 1.1 \kappa \, \d(v), & \qquad & \forall \, v \notin S_{i + 1}.
		\end{alignedat}
		\right. \label{eqn:future-residue-bounds}
	\end{align}
	Here, the first two cases follow because the residue at every active node remains equal to $\lambda \, \d(v)$.
	For the last case, if $v \in S_{a+L} \setminus S_{i + 1}$, then $\r_{a+L}(v) = \lambda \, \d(v)$, whereas \eqref{eqn:residue-bounds} guarantees that $\r_i(v) \le (\lambda + 1.1 \kappa) \, \d(v)$, giving the desired bound.
	If $v \notin S_{a+L}$, then \cref{eqn:frontier-residue-formula} and $\x_i \le \x_{a+L}$ give $\r_i(v) \le \r_{a+L}(v)$.
	This proves \eqref{eqn:future-residue-bounds} in all cases.

	\paragraph{Relating $\Delta \psi(a,a + L)$ to $\Delta \r_i$.}
	Fix $i \in \{a - L,\ldots,a - 1\}$.
	Since $[i,a)$ and $[a,a+L)$ are disjoint, applying \cref{lem:orthogonal} gives
	\begin{align*}
		\sum_{v \in V} \bigl(\x_{a + L}(v) - \x_a(v)\bigr) \bigl(\r_a(v) - \r_i(v)\bigr) = 0.
	\end{align*}
	Combining this with \cref{eqn:interval-objective-decrease} leads to
	\begin{align*}
		\frac{2}{\alpha} \Delta \psi(a,a + L) & = \sum_{v \in V} \bigl(\x_{a+L}(v) - \x_a(v)\bigr) \bigl(\r_a(v) - \r_{a+L}(v)\bigr) \\
		& = \sum_{v \in V} \bigl(\x_{a + L}(v) - \x_a(v)\bigr) \bigl(\r_i(v) - \r_{a + L}(v)\bigr).
	\end{align*}
	Using \eqref{eqn:future-residue-bounds} and the nonnegativity of $\x_{a + L} - \x_a$, we further obtain
	\begin{align}
		\frac{2}{\alpha} \Delta \psi(a,a + L) & \le \sum_{v \in T_i} \bigl(-\Delta \r_i(v)\bigr) \bigl(\x_{a + L}(v) - \x_a(v)\bigr) + \sum_{v \notin S_{i + 1}} 1.1 \kappa \, \d(v) \bigl(\x_{a + L}(v) - \x_a(v)\bigr) \notag \\
		& \le \sum_{v \in T_i} \bigl(-\Delta \r_i(v)\bigr) \bigl(\x_{a + L}(v) - \x_a(v)\bigr) + 1.1 \kappa, \label{eqn:future-decrease-upper-one}
	\end{align}
	where in the last step we used
	\begin{align*}
		\sum_{v \in V} \d(v) \, \bigl(\x_{a + L}(v) - \x_a(v)\bigr) \le \sum_{v \in V} \d(v) \, \x_{a + L}(v) \le 1,
	\end{align*}
	which follows from \cref{eqn:mass-identity,lem:trajectory}.

	\paragraph{Bounding the reserve increase using $\Delta \r_j$.}
	Now additionally fix $j \in \{a - L,\ldots,a - 1\}$.
	By subtracting the PageRank invariant \eqref{eqn:invariant-entrywise} for states $j$ and $a + L$, we have that, for every $v \in V$,
	\begin{align*}
		\x_{a + L}(v) - \x_j(v) = \sum_{u \in V} \bigl(\r_j(u) - \r_{a + L}(u)\bigr) \, \frac{\pr_{\alpha}(u,v)}{\d(v)}.
	\end{align*}
	Now applying \eqref{eqn:future-residue-bounds} gives
	\begin{align*}
		\x_{a + L}(v) - \x_j(v) & \le \sum_{u \in T_j} \bigl(-\Delta \r_j(u)\bigr) \, \frac{\pr_{\alpha}(u,v)}{\d(v)} + 1.1 \kappa \sum_{u \notin S_{j + 1}} \frac{\d(u) \, \pr_{\alpha}(u,v)}{\d(v)} \\
		& \le \sum_{u \in T_j} \bigl(-\Delta \r_j(u)\bigr) \, \frac{\pr_{\alpha}(u,v)}{\d(v)} + 1.1 \kappa,
	\end{align*}
	where the last step used the symmetry identity \eqref{eqn:symmetry} to derive
	\begin{align*}
		\sum_{u \notin S_{j + 1}} \frac{\d(u) \, \pr_{\alpha}(u,v)}{\d(v)} \le \sum_{u \in V} \frac{\d(u) \, \pr_{\alpha}(u,v)}{\d(v)} = \sum_{u \in V} \pr_{\alpha}(v,u) = 1.
	\end{align*}
	Using $\x_{a + L}(v) - \x_a(v) \le \x_{a + L}(v) - \x_j(v)$ and substituting into \eqref{eqn:future-decrease-upper-one}, we obtain
	\begin{align}
		\frac{2}{\alpha} \Delta \psi(a,a + L) & \le \sum_{v \in T_i} \bigl(-\Delta \r_i(v)\bigr) \bigl(\x_{a + L}(v) - \x_j(v)\bigr) + 1.1 \kappa \notag \\
		& \le \sum_{u \in T_j} \sum_{v \in T_i} \frac{1}{\d(v)} \, \bigl(-\Delta \r_j(u)\bigr) \, \pr_{\alpha}(u,v) \, \bigl(-\Delta \r_i(v)\bigr) + 1.1 \kappa \sum_{v \in T_i} \bigl(-\Delta \r_i(v)\bigr) + 1.1 \kappa \notag \\
		& \le \sum_{u \in T_j} \sum_{v \in T_i} \frac{1}{\d(v)} \, \Delta \r_j(u) \, \pr_{\alpha}(u,v) \, \Delta \r_i(v) + 2.2 \kappa, \label{eqn:future-decrease-upper-pair}
	\end{align}
	where in the last step we used $-\Delta \r_i(v) \le \r_i(v)$ on $T_i$ and $\sum_{v \in V} \r_i(v) \le 1$, which follows from \cref{eqn:mass-identity,lem:trajectory}.

	\paragraph{Summing over all pairs.}
	Summing \eqref{eqn:future-decrease-upper-pair} for all $i,j \in \{a - L,\ldots,a - 1\}$ leads to
	\begin{align*}
		\frac{2 L^2}{\alpha} \Delta \psi(a,a + L) \le \sum_{i = a - L}^{a - 1} \sum_{j = a - L}^{a - 1} \sum_{u \in T_j} \sum_{v \in T_i} \frac{1}{\d(v)} \, \Delta \r_j(u) \, \pr_{\alpha}(u,v) \, \Delta \r_i(v) + 2.2 \kappa L^2.
	\end{align*}
	For any $u \in T_j$ and $v \in T_i$, we have
	\begin{align*}
		\frac{1}{\d(v)} \, \Delta \r_j(u) \, \pr_{\alpha}(u,v) \, \Delta \r_i(v) & = \d(u) \, \pr_{\alpha}(u,v) \, \frac{\Delta \r_j(u)}{\d(u)} \, \frac{\Delta \r_i(v)}{\d(v)} \\
		& \le \frac{1}{2} \d(u) \, \pr_{\alpha}(u,v) \, \Bigl(\frac{\Delta \r_j(u)^2}{\d(u)^2} + \frac{\Delta \r_i(v)^2}{\d(v)^2}\Bigr),
	\end{align*}
	where the last step used the AM--GM inequality.
	Substituting this into the preceding inequality leads to
	\begin{align*}
		\frac{2 L^2}{\alpha} \Delta \psi(a,a + L) \le \frac{1}{2} \sum_{i = a - L}^{a - 1} \sum_{j = a - L}^{a - 1} \sum_{u \in T_j} \sum_{v \in T_i} \d(u) \, \pr_{\alpha}(u,v) \, \Bigl(\frac{\Delta \r_j(u)^2}{\d(u)^2} + \frac{\Delta \r_i(v)^2}{\d(v)^2}\Bigr) + 2.2 \kappa L^2.
	\end{align*}

	For the sum involving the first square term above, we have
	\begin{align*}
		& \phantom{{}={}} \sum_{i = a - L}^{a - 1} \sum_{j = a - L}^{a - 1} \sum_{u \in T_j} \sum_{v \in T_i} \d(u) \, \pr_{\alpha}(u,v) \, \frac{\Delta \r_j(u)^2}{\d(u)^2} \\
		& = \sum_{j = a - L}^{a - 1} \sum_{u \in T_j} \frac{\Delta \r_j(u)^2}{\d(u)} \biggl(\sum_{i = a - L}^{a - 1} \sum_{v \in T_i} \pr_{\alpha}(u,v)\biggr) \le \sum_{j = a - L}^{a - 1} \sum_{u \in T_j} \frac{\Delta \r_j(u)^2}{\d(u)},
	\end{align*}
	where the last step used the fact that the expansion sets $T_i$'s are disjoint and $\sum_{v \in V} \pr_{\alpha}(u,v) = 1$.
	On the other hand, for the second square term, we have
	\begin{align*}
		& \phantom{{}={}} \sum_{i = a - L}^{a - 1} \sum_{j = a - L}^{a - 1} \sum_{u \in T_j} \sum_{v \in T_i} \d(u) \, \pr_{\alpha}(u,v) \, \frac{\Delta \r_i(v)^2}{\d(v)^2} \\
		& = \sum_{i = a - L}^{a - 1} \sum_{v \in T_i} \frac{\Delta \r_i(v)^2}{\d(v)^2} \biggl(\sum_{j = a - L}^{a - 1} \sum_{u \in T_j} \d(u) \, \pr_{\alpha}(u,v)\biggr) \\
		& \le \sum_{i = a - L}^{a - 1} \sum_{v \in T_i} \frac{\Delta \r_i(v)^2}{\d(v)^2} \biggl(\sum_{u \in V} \d(u) \, \pr_{\alpha}(u,v)\biggr) = \sum_{i = a - L}^{a - 1} \sum_{v \in T_i} \frac{\Delta \r_i(v)^2}{\d(v)},
	\end{align*}
	where the last step used the symmetry identity \eqref{eqn:symmetry}.
	Note that the two upper bounds above are identical after renaming the indices, so by substituting into the preceding inequality, we obtain
	\begin{align*}
		\frac{2 L^2}{\alpha} \Delta \psi(a,a + L) \le \sum_{i = a - L}^{a - 1} \sum_{v \in T_i} \frac{\Delta \r_i(v)^2}{\d(v)} + 2.2 \kappa L^2.
	\end{align*}
	On the other hand, summing the one-step lower bound \eqref{eqn:delta-psi-lower-bound} on $\Delta \psi_i$ over $i \in \{a-L,\dots,a-1\}$ gives
	\begin{align*}
		\sum_{i = a - L}^{a - 1} \sum_{v \in T_i} \frac{\Delta \r_i(v)^2}{\d(v)} \le \frac{2}{\alpha^2} \, \Delta \psi(a - L,a).
	\end{align*}
	Combining these inequalities and multiplying by $\alpha / \bigl(2 L^2\bigr)$ yield $\Delta \psi(a,a + L) \le \frac{1}{\alpha L^2} \, \Delta \psi(a - L,a) + 1.1 \alpha \kappa$, which finishes the proof.
\end{proof}

Next, the following lemma provides bounds on the objective function that will be used in the iteration-counting argument.

\begin{lemma} \label[lemma]{lem:objective-scale-bounds}
	For every $1 \le b \le K$,
	\begin{align}
		\Delta \psi(0,b) \le \frac{1}{2} \alpha. \label{eqn:total-objective-decrease}
	\end{align}
	Moreover, for every $0 \le i < K$, it holds that
	\begin{align}
		\Delta \psi_i > 0.4 \alpha^2 \eta^2. \label{eqn:objective-decrease-per-iteration}
	\end{align}
\end{lemma}

\begin{proof}
	Since $S_0 = \{s\}$, \eqref{eqn:exact-solution} gives $\d(s) \, \x_0(s) = \alpha \bigl(1 - \lambda \, \d(s)\bigr)$.
	Substituting into \eqref{eqn:objective-quadratic-extension}, we obtain
	\begin{align*}
		\psi_{\lambda}(\x_0) = \frac{1}{2} \d(s) \, \x_0(s)^2 - \alpha \bigl(1 - \lambda \, \d(s)\bigr) \, \x_0(s) = -\frac{1}{2} \d(s) \, \x_0(s)^2 \le 0.
	\end{align*}
	On the other hand, using $\mL_{\alpha} = \alpha \D + (1 - \alpha)(\D - \A)$, for every $\x \ge \vzero$ we have
	\begin{align*}
		\psi_{\lambda}(\x) & \ge \frac{1}{2} \alpha \sum_{v \in V} \d(v) \, \x(v)^2 - \alpha \, \x(s) \\
		& \ge \frac{1}{2} \alpha \, \d(s) \, \x(s)^2 - \alpha \, \x(s) \ge -\frac{\alpha}{2 \d(s)} \ge -\frac{\alpha}{2},
	\end{align*}
	where we used the minimum of a quadratic function in the penultimate step.
	Thus, we have $\psi_{\lambda}(\x_0) - \psi_{\lambda}(\x_b) \le \alpha / 2$ for every $1 \le b \le K$ since $\x_b \ge \vzero$ by \cref{lem:trajectory}.

	For each $0 \le i < K$, the stopping rule is not met at some node $v' \in \partial S_i$, so $\tr_i(v') > (\lambda + \eta) \, \d(v')$ and thus $v' \in T_i$.
	Moreover, \eqref{eqn:approx-frontier-residue-error} and $\r_{i + 1}(v') = \lambda \, \d(v')$ give
	\begin{align*}
		-\Delta \r_i(v') = \r_i(v') - \r_{i + 1}(v') > (\lambda + \eta - 0.1 \kappa) \, \d(v') - \lambda \, \d(v') = (\eta - 0.1 \kappa) \, \d(v').
	\end{align*}
	Using \eqref{eqn:delta-psi-lower-bound}, $\d(v') \ge 1$, and $\frac{1}{2}(\eta - 0.1 \kappa)^2 \ge \frac{1}{2}(0.9\eta)^2 > 0.4 \eta^2$, we have
	\begin{align*}
		\Delta \psi_i \ge \frac{1}{2} \alpha^2 \sum_{v \in T_i} \frac{\Delta \r_i(v)^2}{\d(v)} \ge \frac{1}{2} \alpha^2 \, \frac{\Delta \r_i(v')^2}{\d(v')} > \frac{1}{2} \alpha^2 (\eta - 0.1 \kappa)^2 > 0.4 \alpha^2 \eta^2,
	\end{align*}
	which finishes the proof.
\end{proof}

Using \cref{lem:batch-decrease,lem:objective-scale-bounds}, we are now ready to upper bound the number of iterations in the expansion process.

\begin{lemma} \label[lemma]{lem:number-of-iterations}
	On the event $\mathcal{G}$, the number $K$ of active-set expansions performed by \cref{alg:general-active-set} satisfies
	\begin{align*}
		K < \Bigl\lceil \frac{2}{\sqrt{\alpha}} \Bigr\rceil \Bigl(2 + \log_4\frac{10}{\alpha \eta^2}\Bigr) = O\Bigl(\frac{1}{\sqrt{\alpha}} \log\frac{1}{\alpha \eta}\Bigr).
	\end{align*}
\end{lemma}

\begin{proof}
	Set $L := \bigl\lceil 2 / \sqrt{\alpha} \bigr\rceil$, so $\alpha L^2 \ge 4$ and $L \ge 3$ since $\alpha < 1$.
	If $K < L$, then the claimed bound follows immediately.
	Otherwise, $\lfloor K / L \rfloor \ge 1$.
	By \eqref{eqn:total-objective-decrease}, $\Delta \psi(0,L) \le \alpha / 2$.
	For every $2 \le j \le \lfloor K / L \rfloor$, applying \cref{lem:batch-decrease} with $a = (j - 1)L$ gives
	\begin{align*}
		\Delta \psi\bigl((j - 1)L,jL\bigr) \le \frac{1}{4} \Delta \psi\bigl((j - 2)L,(j - 1)L\bigr) + 1.1 \alpha \kappa.
	\end{align*}
	Iterating this inequality and using $\kappa \le 0.4 \alpha \eta^2$, for every $1 \le j \le \lfloor K / L \rfloor$ we obtain
	\begin{align}
		\Delta \psi\bigl((j - 1)L,jL\bigr) \le \frac{1}{2} \alpha \cdot \frac{1}{4^{j - 1}} + 1.1 \alpha \kappa \cdot \frac{4}{3} < \frac{\alpha}{2 \cdot 4^{j - 1}} + 0.8 \alpha^2 \eta^2. \label{eqn:block-objective-decrease}
	\end{align}

	On the other hand, summing \eqref{eqn:objective-decrease-per-iteration} over iterations $(\lfloor K / L \rfloor - 1)L,\ldots,\lfloor K / L \rfloor L - 1$ gives
	\begin{align*}
		\Delta \psi\Bigl(\bigl(\lfloor K / L \rfloor - 1\bigr)L,\lfloor K / L \rfloor L\Bigr) > 0.4 L \, \alpha^2 \eta^2 \ge 1.2 \alpha^2 \eta^2.
	\end{align*}
	If $\lfloor K / L \rfloor - 1 \ge \log_4\bigl(10 / (\alpha \eta^2)\bigr)$, then applying \eqref{eqn:block-objective-decrease} with $j = \lfloor K / L \rfloor$ gives
	\begin{align*}
		\Delta \psi\Bigl(\bigl(\lfloor K / L \rfloor - 1\bigr)L,\lfloor K / L \rfloor L\Bigr) < (0.05 + 0.8)\alpha^2 \eta^2 < 1.2 \alpha^2 \eta^2,
	\end{align*}
	contradicting this lower bound.
	Therefore, $\lfloor K / L \rfloor - 1 < \log_4\bigl(10 / (\alpha \eta^2)\bigr)$ and thus
	\begin{align*}
		K < \bigl(\lfloor K / L \rfloor + 1\bigr)L < L \Bigl(2 + \log_4\frac{10}{\alpha \eta^2}\Bigr),
	\end{align*}
	as desired.
\end{proof}

\subsection{Establishing the Main Results} \label{sec:main-results}

We are now ready to prove \cref{thm:ACL,thm:regularized} using the properties of \cref{alg:general-active-set}.

\begin{proof}[Proof of \cref{thm:ACL}]
	We use \cref{alg:active-set}, i.e., \cref{alg:general-active-set} with $\lambda = \eta = 0.5 \eps$ and $\kappa = 0.1 \alpha \eps^2$.
	If $\eps \, \d(s) \ge 1$, the algorithm outputs $\tp = \vzero$.
	In this case, its corresponding residue vector $\tr = \e_s$ satisfies $\vzero \le \tr \le \eps \d$, so $\tp$ is an ACL $\eps$-approximate PageRank vector and the theorem follows.

	Now assume $\eps \, \d(s) < 1$.
	By \cref{lem:trajectory}, with probability at least $1 - \delta$, the algorithm terminates with $\tp \ge \vzero$ and $\vzero \le \r\bigl(\D^{-1}\tp\bigr) \le (\lambda + \eta)\d = \eps \d$, so $\tp$ is an ACL $\eps$-approximate PageRank vector.
	Moreover, \cref{lem:trajectory} gives $\vol\bigl(\supp(\tp)\bigr) \le 1 / \lambda = 2 / \eps$.

	It remains to bound the running time on $\mathcal{G}$.
	If the algorithm performs $K$ expansions, it makes $K + 1$ approximate SDD calls.
	By \cref{lem:number-of-iterations} with $\eta = 0.5 \eps$, $K + 1 = O\bigl(\frac{1}{\sqrt{\alpha}} \log\frac{1}{\alpha \eps}\bigr)$.
	Additionally, since each expansion adds at least one node to the active set, we also have $K \le \vol(S_K) = O(1 / \eps)$.
	For every reached active set $S_i$, the matrix $\mL_{\alpha,S_i}$ has dimension at most $\vol(S_i)$ and $O\bigl(\vol(S_i)\bigr)$ nonzero entries.
	Constructing the restricted system and computing the residues on $\partial S_i$ also take $O\bigl(\vol(S_i)\bigr)$ time.
	Using the volume bound in \cref{lem:trajectory}, the solver guarantee in \cref{thm:SDD_solver}, the accuracy and probability parameter settings, we obtain that the total running time is
	\begin{align*}
		O\Bigl((K+1) \, \frac{1}{\eps} \, \polylog\frac{1}{\alpha \eps \delta}\Bigr) = O\Bigl(\min\Bigl\{\frac{1}{\sqrt{\alpha} \, \eps},\frac{1}{\eps^2}\Bigr\} \polylog\frac{1}{\alpha \eps \delta}\Bigr),
	\end{align*}
	which proves the theorem.
\end{proof}

To prove \cref{thm:regularized} for $\ell_1$-regularized PageRank, consider running \cref{alg:general-active-set} with
\begin{align}
	\lambda := \rho, \qquad \eta := \min\Bigl\{\rho,\frac{\xi}{2\alpha}\Bigr\}, \qquad \kappa := 0.4 \alpha \eta^2, \label{eqn:regularized-parameters}
\end{align}
and multiply the output $\tp$ by $\D^{-1}$.
We will use the following facts from \cite{fountoulakis2019variational,wei2026simple}: the unique minimizer satisfies $\xast \ge \vzero$, $\r(\xast) \ge \vzero$, $\vol(\Sast) \le 1 / \rho$, and $\rho \, \d(s) < 1$ implies $s \in \Sast$; moreover, if $s \in S \subseteq \Sast$, then $\x^{(S)} \le \xast$, and every $v \in \partial S$ satisfying $\r\bigl(\x^{(S)}\bigr)(v) > \rho \, \d(v)$ belongs to $\Sast$.

\begin{proof}[Proof of \cref{thm:regularized}]
	The mass identity \eqref{eqn:mass-identity} and $\xast,\r(\xast) \ge \vzero$ give $\d^{\top}\xast \le 1$.
	For $\tx,\xast \ge \vzero$, expanding \eqref{eqn:regularized-objective} gives
	\begin{align*}
		\psi_{\rho}(\tx)-\psi_{\rho}(\xast) & = \frac{1}{2}\bigl(\tx^{\top}\mL_{\alpha} \, \tx - (\xast)^{\top}\mL_{\alpha}\xast\bigr) + \alpha(\rho\d-\e_s)^{\top}(\tx-\xast) \\
		& = \bigl(\mL_{\alpha} \, \tx + \alpha(\rho\d-\e_s)\bigr)^{\top} (\tx-\xast) - \frac{1}{2}\|\tx-\xast\|_{\mL_{\alpha}}^2 \\
		& = \alpha(\rho\d-\tr)^{\top} (\tx-\xast) - \frac{1}{2}\|\tx-\xast\|_{\mL_{\alpha}}^2 \le \alpha(\rho\d-\tr)^{\top}(\tx-\xast),
	\end{align*}
	where the penultimate step uses $\mL_{\alpha} \, \tx = \alpha(\e_s-\tr)$.

	If $(\rho + \eta) \, \d(s) \ge 1$, the algorithm returns $\tx = \vzero$ with residue $\tr = \e_s \le (\rho + \eta)\d \le 2\rho\d$.
	Thus, $\D\tx$ is an ACL $(2\rho)$-approximate PageRank vector.
	We also have
	\begin{align*}
		\psi_{\rho}(\vzero) - \psi_{\rho}(\xast) & \le \alpha (\e_s - \rho\d)^{\top}\xast \le \alpha\eta \, \d^{\top}\xast \le \alpha\eta \le \xi,
	\end{align*}
	and this case takes $O(1)$ time.
	Henceforth, assume $(\rho + \eta) \, \d(s) < 1$, so $\eta \le \rho < 1$ and $\kappa < \eta$.

	By \cref{lem:trajectory}, all SDD calls succeed with probability at least $1 - \delta$.
	Condition on this event $\mathcal{G}$, and let $K$ be the number of expansions.
	The facts from \cite[Lemmas~5.1 and~5.2]{wei2026simple} show inductively that $S_i \subseteq \Sast$: the seed belongs to $\Sast$, and \eqref{eqn:residue-bounds} ensures that every newly activated node satisfies $\r_i(v) > \rho \, \d(v)$ and hence also belongs to $\Sast$.
	Since every expansion adds at least one node, $K + 1 \le |\Sast|$.
	Together with \cref{lem:number-of-iterations}, we have $K+1 = O\bigl(\min\bigl\{|\Sast|,\frac{1}{\sqrt{\alpha}}\log\frac{1}{\alpha\eta}\bigr\}\bigr)$.

	For the output $\tx = \tx_K$ and $\tr := \r(\tx)$, \cref{lem:trajectory} shows that $\vzero \le \tr \le (\rho+\eta)\d \le 2\rho\d$, so $\D\tx$ is an ACL $(2\rho)$-approximate PageRank vector, and $S_K \subseteq \Sast$ gives the claimed support-volume bound.
	Moreover, $\tx,\tr \ge \vzero$ and the mass identity give $\d^{\top}\tx \le 1$, and $\supp(\tx) \subseteq S_K$ and \eqref{eqn:approx-active-residue-error} give $\rho \, \d(v)-\tr(v) \le 0.1\kappa \, \d(v)$ for every $v \in S_K$, while the termination guarantee above gives $\tr-\rho\d \le \eta\d$ on all nodes.
	Using $\tx,\xast \ge \vzero$ to bound the two resulting inner products separately, we obtain
	\begin{align*}
		\psi_{\rho}(\tx)-\psi_{\rho}(\xast) & \le \alpha(\rho\d-\tr)^{\top}\tx + \alpha(\tr-\rho\d)^{\top}\xast \\
		& \le 0.1\alpha\kappa \, \d^{\top}\tx + \alpha\eta \, \d^{\top}\xast \\
		& \le \alpha(\eta+0.1\kappa) \le 2\alpha\eta \le \xi,
	\end{align*}
	where we also used $\d^{\top}\tx,\d^{\top}\xast \le 1$, $\kappa < \eta$, and $\eta \le \xi / (2\alpha)$.

	Finally, each restricted system has $O\bigl(\tvol(\Sast)\bigr)$ nonzero entries, and constructing it and inspecting the boundary take $O\bigl(\vol(\Sast)\bigr)$ time.
	By \cref{thm:SDD_solver} and \eqref{eqn:regularized-parameters}, the total running time is
	\begin{align*}
		O\Bigl((K+1)\Bigl(\tvol(\Sast)\polylog\frac{1}{\alpha\rho\eta\delta}+\vol(\Sast)\Bigr)\Bigr).
	\end{align*}
	Using the two bounds on $K+1$, $\vol(\Sast) \le 1 / \rho$, and the setting of $\eta$ proves all the running-time bounds in \cref{thm:regularized}.
\end{proof}

\subsection{Extension to Sparse Source Vectors} \label{sec:source-distribution}

To support the repeated PageRank computations in \cref{sec:effective-resistance}, we extend \activeSet to allow sparse nonnegative source vectors.
The resulting procedure satisfies the ACL approximation guarantee in \cref{def:ACL}, with running time and support bounds that scale with the source mass.

\begin{lemma} \label[lemma]{lem:sparse-source-pagerank}
	Given a nonzero source vector $\s \ge \vzero$ in sparse representation, $\alpha \in (0,1)$, $0 < \eps \le \|\s\|_1$, and $\delta \in (0,1)$, there exists a procedure $\sparseActiveSet(\s,\alpha,\eps,\delta)$ that returns, with probability at least $1 - \delta$, an ACL $\eps$-approximate PageRank vector $\p$ for source $\s$, together with its corresponding residue vector $\r$.
	On the same event, the outputs satisfy
	\begin{align}
		\|\p\|_1 + \|\r\|_1 = \|\s\|_1, \label{eqn:sparse-source-mass}
	\end{align}
	$\vol\bigl(\supp(\p)\bigr) \le 2\|\s\|_1 / \eps$, $|\supp(\r)| = O\bigl(|\supp(\s)| + \|\s\|_1 / \eps\bigr)$, and the running time is
	\begin{align}
		O\Bigl(|\supp(\s)| + \frac{\|\s\|_1}{\sqrt{\alpha} \, \eps}\polylog \frac{\|\s\|_1}{\alpha \eps \delta} \Bigr). \label{eqn:sparse-source-time}
	\end{align}
\end{lemma}

\begin{proof}
	By normalizing $\s$ and $\eps$ by $\|\s\|_1$ and rescaling the outputs, it suffices to assume $\|\s\|_1 = 1$ and $\eps \le 1$.
	We use \cref{alg:active-set} and replace its early return and initialization by
	\begin{align*}
		S \gets \bigl\{v \in \supp(\s): \s(v) > (0.5 \eps + \kappa) \, \d(v)\bigr\}.
	\end{align*}
	If $S$ is empty, then $\s \le (0.5 \eps + \kappa)\d \le \eps \d$, so we can return $(\vzero,\s)$.
	Otherwise, in each iteration, we use $\alpha (\s_S - 0.5 \eps \, \d_S)$ as the right-hand side vector for the restricted linear system and compute frontier residues by
	\begin{align*}
		\tr(v) \gets \s(v) + \frac{1 - \alpha}{\alpha} \sum_{w \in N(v) \cap S} \tx_S(w), \qquad v \in \partial S.
	\end{align*}
	Additionally, keep the stopping and activation rules unchanged and return $\p := \D\tx$ and $\r := \s - \mL_{\alpha} \, \tx / \alpha$ at termination.

	For $\r(\x) := \s - \mL_{\alpha} \, \x / \alpha$, the mass identity remains $\d^{\top}\x + \vone^{\top}\r(\x) = 1$.
	The proofs of the solver-error bounds and trajectory induction in \cref{lem:active-set-solver-error,lem:trajectory} carry over after replacing $\e_s$ by $\s$, and also $\vol(S) \le 2 / \eps$ holds on successful calls.
	Therefore, the ACL approximation guarantee and the support-volume bound still hold with probability at least $1 - \delta$.
	For the potential-function analysis, replacing the objective's source term by $-\alpha \, \s^{\top}\x$ gives initial objective $-\x_0^{\top}\mL_{\alpha} \, \x_0 / 2 \le 0$ and the lower bound $-\frac{\alpha}{2} \sum_{v \in V} \s(v)^2 / \d(v) \ge -\alpha / 2$ for $\x \ge \vzero$.
	The source cancels in differences of residue invariants, so the proofs of \cref{lem:one-step-objective-decrease,lem:batch-decrease,lem:number-of-iterations} also carry over.

	For the time complexity, initially we need to scan $\supp(\s)$, and then each iteration needs $O\bigl(\vol(S)\bigr)$ work besides the solve.
	The final residue is supported on $\supp(\s) \cup S \cup \partial S$ and thus can be constructed by scanning $\supp(\s)$ and edges incident to $S$, which also yields the desired bound on $|\supp(\r)|$.
	The volume and iteration bounds with \cref{thm:SDD_solver} give the claimed running time.

	Finally, if $\|\s\|_1 \ne 1$, rescaling preserves supports and replaces $1 / \eps$ by $\|\s\|_1 / \eps$, giving \eqref{eqn:sparse-source-time}.
	For the rescaled outputs, the mass identity and nonnegativity also give \eqref{eqn:sparse-source-mass}.
\end{proof}

\section{Applications to Local Clustering and PageRank Estimation} \label{sec:applications}

The ACL approximation guarantee in \cref{thm:ACL} allows our primitive to be used in existing local graph algorithms.
We illustrate this with faster implementations of \PageRankNibble and single-pair PageRank estimation.

\subsection{Faster PageRank-Nibble} \label{sec:PageRank-Nibble}

We consider the problem of local graph clustering.
For $\varnothing \subsetneq S \subsetneq V$, define the conductance of $S$ by
\begin{align*}
	\phi(S) := \frac{\bigl|\bigl\{(u,v) \in E : u \in S, v \in V \setminus S\bigr\}\bigr|}{\min\bigl\{\vol(S),\vol(V\setminus S)\bigr\}}.
\end{align*}
Given a seed node $s$, the \PageRankNibble algorithm of Andersen, Chung, and Lang~\cite{andersen2007pagerank} computes an ACL approximate PageRank vector and examines its sweep sets to find a cluster near $s$ with small conductance.
Its inputs include a seed node $s$, a desired output conductance $\phi$, and a scale parameter $b$.
The original implementation takes $O\bigl(2^b \log^2 n / \phi^2\bigr)$ time (cf. \cite[Algorithm 2, Theorem~6.1]{andersen2007pagerank}) using the local push method.
By replacing the local push method in \PageRankNibble with our faster primitive, we improve the running time to $\tO\bigl(2^b \log^2 n / \phi\bigr)$, as stated below.

\begin{corollary} \label[corollary]{cor:PageRank-Nibble}
	Given a seed node $s$, conductance parameter $\phi \in (0,1)$, an integer $b \ge 1$, and failure probability $\delta \in (0,1)$, there exists an implementation of \PageRankNibble~\cite[Algorithm 2]{andersen2007pagerank} that, with probability at least $1 - \delta$, runs in
	\begin{align*}
		O\Bigl(\frac{2^b \log^2 n}{\phi} \polylog \frac{2^b \log n}{\phi\delta}\Bigr)
	\end{align*}
	time and retains the conductance, volume, and overlap guarantees of~\cite[Theorem~6.2]{andersen2007pagerank} under the same hypotheses.
	In particular, whenever it returns a set $S$, we have $\phi(S) < \phi$ and $\vol(S) > 2^{b - 1}$, so the ratio between its running time and $\vol(S)$ is $\tO(1 / \phi)$.
\end{corollary}

\begin{proof}
	Use the parameters of~\cite[Algorithm~2]{andersen2007pagerank}, which, after converting from lazy to non-lazy walks, satisfy
	\begin{align*}
		\alpha = \Theta\Bigl(\frac{\phi^2}{\log n}\Bigr), \qquad \eps = \Theta\Bigl(\frac{1}{2^b \log n}\Bigr).
	\end{align*}
	Replace the approximate PageRank computation by \cref{thm:ACL}.
	With probability at least $1-\delta$, the resultant vector satisfies the same ACL approximation guarantee, so the proof of~\cite[Theorem~6.2]{andersen2007pagerank} directly applies.
	Plugging in the $O\bigl(\polylog\bigl(1/(\alpha\eps\delta)\bigr) / \bigl(\sqrt{\alpha} \, \eps\bigr)\bigr)$ time complexity gives the desired bound, and the support volume of the approximate PageRank vector is $O(1 / \eps)$, so sorting and examining all sweep sets takes only $O(2^b \log^2 n)$ additional time, giving the stated total running-time bound.
\end{proof}

The $\tO(1 / \phi)$ work-volume ratio of our faster implementation of \PageRankNibble matches that of the evolving-set methods~\cite{andersen2016almost}.
However, evolving-set methods provide stronger conductance guarantees, while our implementation retains the guarantees of the original \PageRankNibble.
We refer the reader to~\cite{andersen2016almost} for a more detailed discussion.

\subsection{Faster Single-Pair PageRank Estimation} \label{sec:single-pair-PageRank}

In the single-pair PageRank estimation problem, we are given a source node $s$, a target node $t$, and a teleportation parameter $\alpha \in (0,1)$, and the goal is to estimate $\pr_{\alpha}(s,t)$ up to additive error $\epsilon \max\bigl\{\pr_{\alpha}(s,t),\theta\bigr\}$.
The Undirected-BiPPR algorithm of Lofgren, Banerjee, and Goel~\cite{lofgren2015bidirectional} combines local push from the source $s$ with random walks from the target $t$.
It first computes an ACL $\eps$-approximate PageRank vector (where $\eps$ is chosen according to the problem parameters), then samples random walks and uses the residues at the endpoints and the symmetry identity \eqref{eqn:symmetry} to estimate the remaining contribution to $\pr_{\alpha}(s,t)$.
For relative error $\epsilon \in (0,1)$, significance threshold $\theta \in (0,1)$, and failure probability $\delta \in (0,1)$, balancing these two phases of local push and random-walk sampling gives a running time of
\begin{align*}
	O\biggl(\frac{1}{\alpha \epsilon} \sqrt{\frac{\d(t) \log(1 / \delta)}{\theta}}\biggr).
\end{align*}
By replacing the local push phase in Undirected-BiPPR by \cref{thm:ACL}, we obtain the following improvement.

\begin{corollary} \label[corollary]{cor:single-pair-PageRank}
	Given $s,t \in V$, $\alpha,\epsilon,\theta \in (0,1)$, and $\delta \in (0,1)$, there is an algorithm that, with probability at least $1 - \delta$, returns an estimate $\widehat{\pr}_{\alpha}(s,t)$ satisfying
	\begin{align*}
		\bigl|\widehat{\pr}_{\alpha}(s,t) - \pr_{\alpha}(s,t)\bigr| \le \epsilon \max\bigl\{\pr_{\alpha}(s,t),\theta\bigr\}
	\end{align*}
	and runs in time
	\begin{align*}
		O\Bigl(\frac{1}{\alpha^{3/4} \, \epsilon} \sqrt{\frac{\d(t)}{\theta}} \, \polylog\frac{\d(t)}{\alpha \epsilon \theta} \log\frac{1}{\delta}\Bigr).
	\end{align*}
\end{corollary}

\begin{proof}
	First consider constant success probability.
	Use \cref{alg:active-set} with accuracy $\eps$ (to be determined later) in place of local push, and form its residue vector using \eqref{eqn:residue-def} in $O(1 / \eps)$ time.
	By \cref{thm:ACL}, the output satisfies the same ACL approximation guarantee required by BiPPR, so the same sampling analysis~\cite[Section~2.4]{lofgren2015bidirectional} applies.
	The sampling phase uses $O\bigl(\eps \, \d(t) / (\epsilon^2 \theta)\bigr)$ random walks, each of expected length $O(1 / \alpha)$.
	With constant probability, the algorithm satisfies the required error guarantee and runs in time
	\begin{align*}
		O\Bigl(\frac{1}{\sqrt{\alpha} \, \eps} \polylog\frac{1}{\alpha \eps} + \frac{\eps \, \d(t)}{\alpha \epsilon^2 \theta}\Bigr).
	\end{align*}
	Balancing the two terms up to polylogarithmic factors by setting
	\begin{align*}
		\eps := \alpha^{1/4} \epsilon \sqrt{\frac{\theta}{\d(t)}}
	\end{align*}
	gives the stated bound for constant failure probability.
	Standard repetition and truncation, taking the median of the returned estimates, give the result for failure probability $\delta$ with an additional factor $\log(1 / \delta)$.
\end{proof}

\section{Application to Effective Resistance Estimation} \label{sec:effective-resistance}

In this section, we present new reductions from effective resistance estimation to PageRank computation using local push, which can be accelerated by our improved primitive.
Our goal is to estimate the effective resistance $R(s,t)$ between two nodes $s$ and $t$ with additive error $\epsilon$.
The result is summarized in the following theorem.

\begin{theorem} \label{thm:effective-resistance}
	Given distinct nodes $s,t \in V$, a lower bound $0 < \gamma \le \nu_2(\D^{-1/2} \, \mL \, \D^{-1/2})$ on the second-smallest eigenvalue of the normalized Laplacian, an accuracy $\epsilon \in (0,1]$, and a failure probability $\delta \in (0,1)$, let
	\begin{align}
		L := \max\Bigl\{\Bigl\lceil \frac{2}{\gamma} \ln\Bigl(\frac{2}{\epsilon \gamma} \Bigl(\frac{1}{\d(s)} + \frac{1}{\d(t)}\Bigr)\Bigr)\Bigr\rceil,1\Bigr\}. \label{eqn:resistance-length}
	\end{align}
	Then there exists a randomized algorithm that, with probability at least $1 - \delta$, returns an estimate $\hat{R}$ satisfying $\bigl|\hat{R} - R(s,t)\bigr| \le \epsilon$ in time
	\begin{align*}
		O\biggl(\min\biggl\{\frac{L^{4/3}}{\epsilon^{2/3}},\frac{L^{7/4}}{\epsilon \, \sqrt{\min\{\d(s),\d(t)\}}}\biggr\} \polylog\frac{L}{\epsilon} \log \frac{1}{\delta}\biggr).
	\end{align*}
\end{theorem}

As in \cite{cui2025mixing,yang2025improved,kwok2026solving}, our result assumes that a lower bound $\gamma$ on the spectral gap is given.
Also, our choice of $L$ agrees up to constant factors with the corresponding truncation lengths in these works.
We compare our results with the following bounds in \cite[Corollary~10]{kwok2026solving}, which subsumes the corresponding bounds in \cite{cui2025mixing,yang2025improved}:
\begin{align*}
	O\biggl(\min\biggl\{\frac{L^{7/3}}{\epsilon^{2/3}},\frac{L^{5/2}}{\epsilon \, \sqrt{\min\{\d(s),\d(t)\}}}\biggr\} \log\frac{1}{\delta}\biggr).
\end{align*}
Our results improve the respective exponents of $L$ from $7/3$ to $4/3$ and from $5/2$ to $7/4$.
In fact, as we show at the end of this section, combining our reduction with the original local push algorithm gives complexity bounds of
\begin{align}
	O\biggl(\min\biggl\{\frac{L^{5/3}}{\epsilon^{2/3}},\frac{L^2}{\epsilon \, \sqrt{\min\{\d(s),\d(t)\}}}\biggr\} \polylog\frac{L}{\epsilon} \log\frac{1}{\delta}\biggr), \label{eqn:resistance-push-time}
\end{align}
so the reduction itself improves both terms, and \activeSet further accelerates it.

\subsection{Reduction from Effective Resistance to PageRank} \label{sec:resistance-properties} \label{sec:resistance-reduction}

We first relate effective resistance to PageRank powers.
To this end, we define the PageRank matrix
\begin{align}
	\mat{\Pi}_{\alpha} := \alpha \bigl(\I - (1 - \alpha) \A \D^{-1}\bigr)^{-1}. \label{eqn:pagerank-matrix}
\end{align}
It holds that $\mat{\Pi}_{\alpha} \, \s = \pr_{\alpha}(\s)$ for any source vector $\s \in \R^V$.

The following lemma gives a PageRank representation of effective resistance and bounds its truncation error.
\begin{lemma} \label[lemma]{lem:resistance-expansion}
	Let $s,t,\gamma,\epsilon,L$ be as in \cref{thm:effective-resistance}.
	For every $\alpha \in (0,1)$, we have
	\begin{align}
		R(s,t) = \frac{1 - \alpha}{\alpha} \sum_{i = 1}^{\infty} (\e_s - \e_t)^{\top} \D^{-1} \, \mat{\Pi}_{\alpha}^i (\e_s - \e_t). \label{eqn:resistance-expansion}
	\end{align}
	Moreover, it holds that
	\begin{align}
		0 \le \sum_{i = L + 1}^{\infty} (\e_s - \e_t)^{\top} \D^{-1} \, \mat{\Pi}_{1/2}^i (\e_s - \e_t) \le \frac{1}{2} \epsilon \label{eqn:resistance-tail}
	\end{align}
	and
	\begin{align}
		R(s,t) \le \frac{2L}{\min\{\d(s),\d(t)\}} + \frac{1}{2} \epsilon. \label{eqn:resistance-degree-bound}
	\end{align}
\end{lemma}

\begin{proof}
	Fix $\alpha \in (0,1)$.
	We will show that the vector
	\begin{align*}
		\x := \frac{1 - \alpha}{\alpha} \, \D^{-1} \sum_{i = 1}^{\infty} \mat{\Pi}_{\alpha}^i (\e_s - \e_t)
	\end{align*}
	satisfies $\mL \x = \e_s - \e_t$.
	As $G$ is connected, any such vector $\x$ differs from $\mL^+ (\e_s - \e_t)$ only by a constant vector, whose inner product with $\e_s - \e_t$ is zero.
	Therefore,
	\begin{align*}
		R(s,t) = (\e_s - \e_t)^{\top} \mL^+ (\e_s - \e_t) = (\e_s - \e_t)^{\top} \x,
	\end{align*}
	and substituting the series for $\x$ gives \eqref{eqn:resistance-expansion}.

	We first prove absolute convergence of the series defining $\x$ by comparison with a geometric series.
	By the definition of $\mat{\Pi}_{\alpha}$ and $\mL = \D - \A$, we have
	\begin{align*}
		\mat{\Pi}_{\alpha}^{-1} & = \frac{1}{\alpha} \I - \frac{1 - \alpha}{\alpha} \, \A \D^{-1} = \I + \frac{1 - \alpha}{\alpha} \, \mL \D^{-1} \\
		& = \D^{1/2} \Bigl(\I + \frac{1 - \alpha}{\alpha} \, \D^{-1/2} \, \mL \, \D^{-1/2}\Bigr) \D^{-1/2}.
	\end{align*}
	Taking inverses and powers gives, for every integer $i \ge 1$,
	\begin{align*}
		\mat{\Pi}_{\alpha}^{i} = \D^{1/2} \Bigl(\I + \frac{1 - \alpha}{\alpha} \, \D^{-1/2} \, \mL \, \D^{-1/2}\Bigr)^{-i} \, \D^{-1/2}.
	\end{align*}
	Note that the eigenspace of the normalized Laplacian corresponding to eigenvalue zero is spanned by $\D^{1/2} \vone$, and $(\D^{1/2} \vone)^{\top} \D^{-1/2} (\e_s - \e_t) = 0$.
	Thus, when $\mat{\Pi}_{\alpha}^{i}$ is applied to $\e_s - \e_t$, the middle factor acts only on eigendirections with $\nu \ge \gamma$, scaling each by $\bigl(\alpha / (\alpha + (1 - \alpha)\nu)\bigr)^i$.
	Since $0 < \alpha < 1$ and $\gamma > 0$, we obtain
	\begin{align*}
		& \phantom{{}\le{}} \sum_{i = 1}^{\infty} \Bigl\|\frac{1 - \alpha}{\alpha} \, \D^{-1} \, \mat{\Pi}_{\alpha}^{i} (\e_s - \e_t)\Bigr\|_2 \\
		& \le \frac{1 - \alpha}{\alpha} \, \bigl\|\D^{-1/2}\bigr\|_2 \, \bigl\|\D^{-1/2} (\e_s - \e_t)\bigr\|_2 \sum_{i = 1}^{\infty} \Bigl(\frac{\alpha}{\alpha + (1 - \alpha)\gamma}\Bigr)^i < \infty.
	\end{align*}
	Thus the series defining $\x$ converges absolutely.
	We next verify that $\x$ satisfies $\mL \x = \e_s - \e_t$.
	Using the expression for $\mat{\Pi}_{\alpha}^{-1}$ above, we have
	\begin{align*}
		\frac{1 - \alpha}{\alpha} \, \mL \D^{-1} \, \mat{\Pi}_{\alpha} = \bigl(\mat{\Pi}_{\alpha}^{-1} - \I\bigr) \mat{\Pi}_{\alpha} = \I - \mat{\Pi}_{\alpha}.
	\end{align*}
	Hence, for every integer $i \ge 1$, expanding the sum and canceling consecutive terms gives
	\begin{align*}
		\mL \biggl(\frac{1 - \alpha}{\alpha} \D^{-1}\sum_{j = 1}^{i} \mat{\Pi}_{\alpha}^j (\e_s - \e_t)\biggr) = \sum_{j = 1}^{i} \bigl(\mat{\Pi}_{\alpha}^{j - 1} - \mat{\Pi}_{\alpha}^j\bigr) (\e_s - \e_t) = \bigl(\I - \mat{\Pi}_{\alpha}^{i}\bigr) (\e_s - \e_t).
	\end{align*}
	Taking $i \to \infty$ and using $\mat{\Pi}_{\alpha}^{i} (\e_s - \e_t) \to \vzero$ proves $\mL \x = \e_s - \e_t$.

	For the rest of the statements, take $\alpha = 1/2$.
	The eigenvalue of $\sum_{i = L + 1}^{\infty} \bigl(\I + \D^{-1/2} \, \mL \, \D^{-1/2}\bigr)^{-i}$, restricted to a normalized Laplacian eigenspace with eigenvalue $\nu > 0$, is $\sum_{i = L + 1}^{\infty} (1 + \nu)^{-i} = 1 / \bigl(\nu(1 + \nu)^L\bigr)$.
	Since every positive eigenvalue of the normalized Laplacian is at least $\gamma$, we have
	\begin{align*}
		& \phantom{{}={}} \sum_{i = L + 1}^{\infty} (\e_s - \e_t)^{\top} \D^{-1} \, \mat{\Pi}_{1/2}^i (\e_s - \e_t) \\
		& \le \frac{1}{\gamma} \Bigl(\frac{1}{\d(s)} + \frac{1}{\d(t)}\Bigr) (1 + \gamma)^{-L} \le \frac{1}{\gamma} \Bigl(\frac{1}{\d(s)} + \frac{1}{\d(t)}\Bigr) e^{-\gamma L / 2} \le \frac{1}{2} \epsilon,
	\end{align*}
	using $\ln(1 + \gamma) \ge \gamma / 2$ for $0 < \gamma \le 2$ and \cref{eqn:resistance-length}.
	The same spectral representation gives, for every $i \ge 1$,
	\begin{align*}
		0 & \le (\e_s - \e_t)^{\top} \D^{-1} \, \mat{\Pi}_{1/2}^i (\e_s - \e_t) \\
		& \le \bigl\|\D^{-1/2} (\e_s - \e_t)\bigr\|_2^2 = \frac{1}{\d(s)} + \frac{1}{\d(t)} \le \frac{2}{\min\{\d(s),\d(t)\}}.
	\end{align*}
	This also proves the nonnegativity in \eqref{eqn:resistance-tail}.
	Bounding the first $L$ terms and applying \eqref{eqn:resistance-tail} yields \eqref{eqn:resistance-degree-bound}.
\end{proof}

We may therefore return zero if $\min\bigl\{\d(s),\d(t)\bigr\} \ge 4L / \epsilon$.
Now set
\begin{align*}
	\alpha := \frac{1}{L + 1}, \qquad \ell := \Bigl\lceil \log_2\frac{8L}{\epsilon}\Bigr\rceil.
\end{align*}
The following lemma reduces resistance estimation to three entries of a sum of PageRank powers.
\begin{lemma} \label[lemma]{lem:resistance-pagerank}
	For $u,v \in V$, define
	\begin{align}
		F(u,v) := \frac{1}{\d(v)} \biggl(\sum_{i = 1}^\ell \mat{\Pi}_{\alpha}^i \e_u\biggr)(v). \label{eqn:resistance-F}
	\end{align}
	Then $0 \le R(s,t) - L \bigl(F(s,s) + F(t,t) - 2F(s,t)\bigr) \le \frac{3}{4} \epsilon$.
\end{lemma}

\begin{proof}
	The matrix $\D^{-1/2} \, \mat{\Pi}_{\alpha} \, \D^{1/2} = \bigl(\I + L \, \D^{-1/2} \, \mL \, \D^{-1/2}\bigr)^{-1}$ is symmetric, as is each of its powers.
	The similarity transformation gives
	\begin{align*}
		\D^{-1} \, \mat{\Pi}_{\alpha}^i = \D^{-1/2} \bigl(\D^{-1/2} \, \mat{\Pi}_{\alpha} \, \D^{1/2}\bigr)^i \D^{-1/2},
	\end{align*}
	which is therefore symmetric for every integer $i \ge 0$.
	It follows that
	\begin{align}
		\frac{(\mat{\Pi}_{\alpha}^i \e_u)(v)}{\d(v)} = \frac{(\mat{\Pi}_{\alpha}^i \e_v)(u)}{\d(u)}, \qquad \forall \, u,v \in V,\ i \ge 0. \label{eqn:pagerank-power-symmetry}
	\end{align}
	Thus $F(s,t) = F(t,s)$.
	To bound the truncation error, we first establish the scalar bound
	\begin{align}
		L \sum_{i = \ell + 1}^{\infty} (1 + L \nu)^{-i} = \frac{(1 + L \nu)^{-\ell}}{\nu} \le \frac{(1 + \nu)^{-L}}{\nu} + L \, 2^{-\ell}, \qquad \forall \, \nu > 0. \label{eqn:eigenvalue-scalar-bound}
	\end{align}
	For $\nu > 1 / L$, the left-hand side is less than $L \, 2^{-\ell}$.
	For $0 < \nu \le 1 / L$, using $\ln(1 + L \nu) \ge L \nu / 2$, $\ln(1 + \nu) \le \nu$, and $\ell \ge 2$, we have
	\begin{align*}
		\ell \ln(1 + L \nu) \ge \frac{1}{2} \ell L \nu \ge L \nu \ge L \ln(1 + \nu),
	\end{align*}
	giving $(1 + L \nu)^{-\ell} \le (1 + \nu)^{-L}$.
	So \eqref{eqn:eigenvalue-scalar-bound} holds for all $\nu > 0$.

	Let $\vec{u}_1,\ldots,\vec{u}_n$ be an orthonormal eigenbasis of $\D^{-1/2} \, \mL \, \D^{-1/2}$, with corresponding eigenvalues $0 = \nu_1 < \nu_2 \le \cdots \le \nu_n$.
	Since $\D^{-1/2} (\e_s - \e_t)$ is orthogonal to $\vec{u}_1$, applying \cref{eqn:resistance-expansion} with $\alpha = 1 / (L + 1)$ and hence $(1 - \alpha) / \alpha = L$ gives
	\begin{align*}
		& \phantom{{}={}} R(s,t) - L \bigl(F(s,s) + F(t,t) - 2F(s,t)\bigr) \\
		& = L \sum_{i = \ell + 1}^{\infty} (\e_s - \e_t)^{\top} \D^{-1} \, \mat{\Pi}_{\alpha}^i (\e_s - \e_t) \\
		& = \sum_{i = 2}^{n} \bigl(\vec{u}_i^{\top} \D^{-1/2} (\e_s - \e_t)\bigr)^2 \, \frac{(1 + L \nu_i)^{-\ell}}{\nu_i}.
	\end{align*}
	Here the second equality follows by expanding $\D^{-1/2} (\e_s - \e_t)$ in the orthonormal eigenbasis and using $L \sum_{j = \ell + 1}^{\infty} (1 + L \nu_i)^{-j} = (1 + L \nu_i)^{-\ell} / \nu_i$ for each $i \ge 2$.
	Applying \eqref{eqn:eigenvalue-scalar-bound}, orthonormality, and the identity $\sum_{j = L + 1}^{\infty} (1 + \nu_i)^{-j} = (1 + \nu_i)^{-L} / \nu_i$, we obtain
	\begin{align*}
		& \phantom{{}={}} R(s,t) - L \bigl(F(s,s) + F(t,t) - 2F(s,t)\bigr) \\
		& \le \sum_{i = L + 1}^{\infty} (\e_s - \e_t)^{\top} \D^{-1} \, \mat{\Pi}_{1/2}^i (\e_s - \e_t) + L \, 2^{-\ell} \bigl\|\D^{-1/2} (\e_s - \e_t)\bigr\|_2^2 \\
		& \le \frac{1}{2} \epsilon + L \, 2^{-\ell} \Bigl(\frac{1}{\d(s)} + \frac{1}{\d(t)}\Bigr) \le \frac{1}{2} \epsilon + 2L \, 2^{-\ell} \le \frac{3}{4} \epsilon,
	\end{align*}
	which finishes the proof.
\end{proof}

\subsection{Approximating PageRank Powers}

We next design an algorithm to approximate the sum of PageRank powers $\sum_{i = 1}^\ell \mat{\Pi}_{\alpha}^i \e_u$ that appears in \cref{eqn:resistance-F}.
\Cref{alg:accumulate} does so by repeatedly invoking \sparseActiveSet using the current reserve vector plus a unit vector as the source vector.
The next lemma characterizes the properties of this procedure.

\begin{algorithm}[ht]
	\DontPrintSemicolon
	\caption{$\accumulate(u,\alpha,\ell,\theta,\delta)$} \label{alg:accumulate}
	\KwIn{source node $u \in V$, teleportation parameter $\alpha \in (0,1)$, integer $\ell \ge 1$, threshold $\theta \in (0,1]$, failure probability $\delta \in (0,1)$}
	\KwOut{reserve $\p_\ell$ and residues $\r_1,\ldots,\r_\ell$}
	$\p_0 \gets \vzero$ \;
	\For{$i \gets 1$ \KwTo $\ell$}
	{
		$(\p_i,\r_i) \gets \sparseActiveSet(\e_u + \p_{i - 1},\alpha,\theta,\delta / \ell)$ \;
	}
	\Return{$\p_\ell,\r_1,\ldots,\r_\ell$} \;
\end{algorithm}

\begin{lemma} \label[lemma]{lem:accumulate}
	With probability at least $1 - \delta$, $\accumulate(u,\alpha,\ell,\theta,\delta)$ returns nonnegative vectors $\p_\ell,\r_1,\ldots,\r_\ell$ satisfying
	\begin{align}
		\sum_{i = 1}^\ell \mat{\Pi}_{\alpha}^i \e_u = \p_\ell + \sum_{i = 1}^\ell \mat{\Pi}_{\alpha}^{\ell - i + 1} \r_i. \label{eqn:accumulated-residue}
	\end{align}
	On the same event, the residues satisfy $\vzero \le \r_i \le \theta \d$ for every $i \in \{1,\ldots,\ell\}$.
	Moreover, its running time is $\tO\bigl(\ell^2 / \bigl(\sqrt{\alpha} \, \theta\bigr)\bigr)$, $\vol\bigl(\supp(\p_\ell)\bigr) \le 2\ell / \theta$, and the retained residues have $O(\ell^2 / \theta)$ nonzero entries in total.
\end{lemma}

\begin{proof}
	If the first $i - 1$ calls succeed, induction using \eqref{eqn:sparse-source-mass} gives $1 \le \|\e_u + \p_{i - 1}\|_1 = 1 + \|\p_{i - 1}\|_1 \le i$.
	Thus, the chosen threshold $\theta \in (0,1]$ satisfies the requirement $0 < \theta \le \|\e_u + \p_{i - 1}\|_1$ in \cref{lem:sparse-source-pagerank}.
	Since the $i$-th call fails with probability at most $\delta / \ell$ given that the first $i - 1$ calls succeed, a union bound over the $\ell$ possible positions of the first failure shows that all calls succeed with probability at least $1 - \delta$.
	On this event, $\p_i = \pr_{\alpha}(\e_u+\p_{i-1}-\r_i)$ gives
	\begin{align*}
		\p_i = \mat{\Pi}_{\alpha} \, \e_u - \mat{\Pi}_{\alpha} \r_i + \mat{\Pi}_{\alpha} \, \p_{i - 1}.
	\end{align*}
	Collectively, this identity for $i = 1,\ldots,\ell$ along with $\p_0 = \vzero$ leads to \cref{eqn:accumulated-residue}.
	The mass bound and \cref{lem:sparse-source-pagerank} give $\vol\bigl(\supp(\p_i)\bigr) \le 2i / \theta$.
	Since $\d(v) \ge 1$, the $i$-th source and residue each have $O(i / \theta)$ nonzero entries, and the call takes $\tO\bigl(i / \bigl(\sqrt{\alpha} \, \theta\bigr)\bigr)$ time.
	Summing over $i$ proves the time and storage bounds, including sparse source construction.
\end{proof}

We remark that the original local push algorithm gives the same guarantees deterministically with running time $O\bigl(\ell^2 / (\alpha \theta)\bigr)$.

\subsection{Bidirectional Estimation} \label{sec:resistance-bidirectional}

After running \accumulate, we estimate the contributions of the residues in \cref{eqn:accumulated-residue} by sampling random walks.
This bidirectional approach is similar to the one used in \cite{lofgren2016personalized,lofgren2015bidirectional,cui2025mixing,yang2025improved,kwok2026solving}.

Specifically, for each $u \in \{s,t\}$, let $\p^{(u)}$ denote the final reserve $\p_\ell$ returned by \accumulate from $u$, and let $\r_1^{(u)},\ldots,\r_\ell^{(u)}$ denote its retained residues.
For each pair $(u,v) \in \bigl\{(s,s),(t,t),(s,t)\bigr\}$, sample $I$ uniformly from $\{1,\ldots,\ell\}$ and run $\ell - I + 1$ successive $\alpha$-terminated walks using independent randomness, starting from $v$ and then from each preceding endpoint.
Here, an $\alpha$-terminated walk terminates with probability $\alpha$ before each step and otherwise moves to a uniformly random neighbor.
At the final endpoint $w$, let
\begin{align}
	Y(u,v) := \ell \, \frac{\r_I^{(u)}(w)}{\d(w)}. \label{eqn:resistance-reverse-sample}
\end{align}

\begin{lemma} \label[lemma]{lem:resistance-reverse-sample}
	Given nonnegative accumulation outputs satisfying the properties given in \cref{lem:accumulate}, the sample $Y(u,v)$ satisfies the following for every pair $(u,v) \in \bigl\{(s,s),(t,t),(s,t)\bigr\}$:
	\begin{align}
		& 0 \le Y(u,v) \le \ell \theta, \notag \\
		& \E\bigl[Y(u,v)\bigr] = F(u,v) - \frac{\p^{(u)}(v)}{\d(v)}, \label{eqn:resistance-sample-mean} \\
		& \Var\bigl[Y(u,v)\bigr] \le \ell^2 \min\Bigl\{\theta^2,\frac{\theta}{\min\{\d(s),\d(t)\}}\Bigr\}. \label{eqn:resistance-sample-variance}
	\end{align}
\end{lemma}

\begin{proof}
	Conditioned on $I = i$, the endpoint has distribution $\mat{\Pi}_{\alpha}^{\ell - i + 1} \e_v$.
	By \cref{eqn:pagerank-power-symmetry},
	\begin{align*}
		\E\bigl[Y(u,v)\bigr] = \sum_{i = 1}^\ell \sum_{x \in V} \bigl(\mat{\Pi}_{\alpha}^{\ell - i + 1} \e_v\bigr)(x) \, \frac{\r_i^{(u)}(x)}{\d(x)} = \frac{1}{\d(v)} \sum_{i = 1}^\ell \bigl(\mat{\Pi}_{\alpha}^{\ell - i + 1} \r_i^{(u)}\bigr)(v),
	\end{align*}
	which equals $F(u,v) - \p^{(u)}(v) / \d(v)$ by \cref{eqn:accumulated-residue}; the range of $Y(u,v)$ follows from $\vzero \le \r_i^{(u)} \le \theta \d$.
	Since $\mat{\Pi}_{\alpha}$ is column stochastic and the reserves are nonnegative, $\E\bigl[Y(u,v)\bigr] \le F(u,v) \le \ell / \d(v) \le \ell / \min\bigl\{\d(s),\d(t)\bigr\}$.
	Thus we have
	\begin{align*}
		\Var\bigl[Y(u,v)\bigr] \le \E\bigl[Y(u,v)^2\bigr] \le \ell \theta \, \E\bigl[Y(u,v)\bigr] \le \ell^2 \min\Bigl\{\theta^2,\frac{\theta}{\min\{\d(s),\d(t)\}}\Bigr\},
	\end{align*}
	where we also used $0 \le Y(u,v) \le \ell\theta$.
	This finishes the proof.
\end{proof}

We present the complete effective resistance estimation algorithm using \activeSet with constant failure probability in \cref{alg:effective-resistance}.
With this algorithm, we now prove the results for effective resistance estimation.

\begin{algorithm}[ht]
	\DontPrintSemicolon
	\caption{$\textup{\texttt{EffectiveResistance}}(s,t,L,\epsilon)$} \label{alg:effective-resistance}
	\KwIn{distinct nodes $s,t \in V$, integer $L$ chosen by \eqref{eqn:resistance-length}, accuracy $\epsilon \in (0,1]$}
	\KwOut{an estimate $\hat{R}$ of $R(s,t)$}
	\If{$\min\{\d(s),\d(t)\} \ge 4L / \epsilon$}
	{
		\Return{$0$} \;
	}
	$\alpha \gets 1 / (L + 1),\ell \gets \bigl\lceil \log_2(8L / \epsilon) \bigr\rceil$ \;
	$\theta \gets \frac{1}{2} \max\bigl\{\epsilon^{2/3} L^{-5/6},\epsilon \, \sqrt{\min\{\d(s),\d(t)\}} \, L^{-5/4}\bigr\}$ \;
	$\bigl(\p^{(s)},\r_1^{(s)},\dots,\r_\ell^{(s)}\bigr) \gets \accumulate(s,\alpha,\ell,\theta,1/16)$ \;
	$\bigl(\p^{(t)},\r_1^{(t)},\dots,\r_\ell^{(t)}\bigr) \gets \accumulate(t,\alpha,\ell,\theta,1/16)$ \;
	$\hat{R} \gets L \Bigl(\frac{\p^{(s)}(s)}{\d(s)} + \frac{\p^{(t)}(t)}{\d(t)} - 2 \, \frac{\p^{(s)}(t)}{\d(t)}\Bigr)$ \;
	$N \gets \Bigl\lceil \frac{10^5L^2\ell^2}{\epsilon^2} \min\bigl\{\theta^2,\theta / \min\{\d(s),\d(t)\}\bigr\}\Bigr\rceil$ \;
	\ForEach{$(u,v) \in \{(s,s),(t,t),(s,t)\}$}
	{
		generate $N$ independent samples $Y_1(u,v),\ldots,Y_N(u,v)$ by \eqref{eqn:resistance-reverse-sample} \;
	}
	$\hat{R} \gets \hat{R} + \frac{L}{N} \sum_{i = 1}^N \bigl(Y_i(s,s) + Y_i(t,t) - 2Y_i(s,t)\bigr)$ \;
	\Return{$\hat{R}$} \;
\end{algorithm}

\begin{proof}[Proof of \cref{thm:effective-resistance} and \eqref{eqn:resistance-push-time}]
	The early return in \cref{alg:effective-resistance} is correct by \eqref{eqn:resistance-degree-bound}.
	Otherwise, $\min\bigl\{\d(s),\d(t)\bigr\} < 4L / \epsilon$.
	This condition ensures that the threshold $\theta$ in \cref{alg:effective-resistance} lies in $(0,1]$, as required by \accumulate.
	By \cref{lem:accumulate}, both accumulation calls succeed with probability at least $7/8$ for \activeSet, and deterministically for the original local push method.

	\paragraph{Accuracy.}
	For fixed successful accumulation outputs, \cref{eqn:resistance-sample-mean} gives $\E[\hat{R}] = L \bigl(F(s,s) + F(t,t) - 2F(s,t)\bigr)$.
	Independence and \eqref{eqn:resistance-sample-variance} give
	\begin{align*}
		\Var[\hat{R}] & = \frac{L^2}{N} \Bigl(\Var\bigl[Y(s,s)\bigr] + \Var\bigl[Y(t,t)\bigr] + 4 \Var\bigl[Y(s,t)\bigr]\Bigr) \\
		& \le \frac{6L^2\ell^2}{N} \min\Bigl\{\theta^2,\frac{\theta}{\min\{\d(s),\d(t)\}}\Bigr\}.
	\end{align*}
	Chebyshev's inequality then gives
	\begin{align*}
		\Pr\Bigl[\bigl|\hat{R} - \E[\hat{R}]\bigr| > \frac{1}{4} \epsilon\Bigr] \le \frac{100L^2\ell^2}{N \epsilon^2} \min\Bigl\{\theta^2,\frac{\theta}{\min\{\d(s),\d(t)\}}\Bigr\} \le \frac{1}{16}.
	\end{align*}
	Together with \cref{lem:resistance-pagerank}, this gives error at most $\frac{3}{4} \epsilon + \frac{1}{4} \epsilon = \epsilon$ with conditional probability at least $15/16$.

	\paragraph{Running time.}
	The sampling process takes $O(N \ell / \alpha) = O(N \ell L)$ expected time.
	By standard repetition and truncation operations and adjusting the constants, we can convert this to a worst-case bound of $O\bigl(N \ell L \log(N\ell)\bigr)$.
	Substituting the choice of $N$ and using $\ell = O(\log(L / \epsilon))$ gives the total running-time bound
	\begin{align*}
		T_{\mathrm{acc}} + \tO\biggl(L + \frac{L^3}{\epsilon^2} \min\biggl\{\theta^2,\frac{\theta}{\min\{\d(s),\d(t)\}}\biggr\}\biggr),
	\end{align*}
	where $T_{\mathrm{acc}}$ is the cost of the two invocations of \accumulate.

	With \activeSet, \cref{lem:accumulate} and $1 / \alpha = L + 1$ give $T_{\mathrm{acc}} = \tO\bigl(\sqrt{L} / \theta\bigr)$.
	Balancing the accumulation and sampling costs gives the threshold used in \cref{alg:effective-resistance} and the total running-time bound
	\begin{align*}
		\tO\biggl(\min\biggl\{\frac{L^{4/3}}{\epsilon^{2/3}},\frac{L^{7/4}}{\epsilon \, \sqrt{\min\{\d(s),\d(t)\}}}\biggr\}\biggr).
	\end{align*}
	Here the additive $L$ term can be absorbed because $\min\bigl\{\d(s),\d(t)\bigr\} < 4L / \epsilon$ makes both displayed terms $\Omega(L)$.

	If \accumulate instead uses the original local push algorithm, then $T_{\mathrm{acc}} = \tO(L / \theta)$.
	The same balancing argument (with a different setting of the parameter $\theta$) gives the total running-time bound
	\begin{align*}
		\tO\biggl(\min\biggl\{\frac{L^{5/3}}{\epsilon^{2/3}},\frac{L^2}{\epsilon \, \sqrt{\min\{\d(s),\d(t)\}}}\biggr\}\biggr).
	\end{align*}
	This proves \eqref{eqn:resistance-push-time}, showing that the reduction itself yields a speedup even with the original local push algorithm.

	The accumulation, sampling-accuracy, and truncation failure probabilities sum to at most $1/8 + 1/16 + 1/16 = 1/4$, so the accuracy and running-time bounds hold jointly with probability at least $3/4$.
	Using standard truncation and the median trick reduces the failure probability to $\delta$ with an $O\bigl(\log(1 / \delta)\bigr)$ multiplicative overhead in running time.
\end{proof}

\section{Acknowledgments and AI Disclosure}

This research was supported in part by National Natural Science Foundation of China (No. U2241212).

The initial proofs of our results were developed through extensive interactions with GPT-5.6 Sol and GPT-6 Astra, with the latter used for the reduction from effective resistance estimation to PageRank computation.
The authors have verified, significantly simplified, and rewritten the proofs to improve readability and presentation.

\printbibliography

\end{document}